\documentclass[review]{IEEEtran}

\usepackage{amsmath,amsfonts,bm}

\def\eqref#1{equation~\ref{#1}}

\def\1{\bm{1}}

\def\vc{{\bm{c}}}

\def\vu{{\bm{u}}}
\def\vv{{\bm{v}}}

\def\vx{{\bm{x}}}
\def\vy{{\bm{y}}}
\def\vz{{\bm{z}}}

\DeclareMathAlphabet{\mathsfit}{\encodingdefault}{\sfdefault}{m}{sl}
\SetMathAlphabet{\mathsfit}{bold}{\encodingdefault}{\sfdefault}{bx}{n}

\usepackage{lineno, hyperref}
\modulolinenumbers[5]

\usepackage{booktabs} 
\usepackage{bibnames}
\usepackage{tablefootnote}

\newcommand{\reals}{\mathbb{R}}

\usepackage[ruled, vlined, linesnumbered]{algorithm2e}
\usepackage{amsmath,bm,mathtools}
\usepackage{amssymb}
\usepackage{amsfonts}
\usepackage{graphicx}
\usepackage{subfig}
\usepackage{xcolor}
\usepackage{multirow, adjustbox}
\usepackage{enumitem}
\usepackage{algorithmic}
\usepackage{todonotes}
\usepackage{threeparttable, adjustbox,booktabs, multirow}
\usepackage{stackengine}
\usepackage[normalem]{ulem}
\usepackage{amsthm}
\newtheorem{theorem}{Theorem}
\newtheorem{definition}{Definition}
\newtheorem{lemma}{Lemma}

\newtheorem{corollary}{Corollary}

\newtheorem{example}{Example}

\allowdisplaybreaks

\newcommand{\relu}{\text{ReLU}}
\newcommand{\change}[1]{ #1}

\newcommand{\chen}[1]{{{\color{red} \textbf{(Chen: #1)}}}}

\newcommand{\zhou}[1]{{#1}}

\SetCommentSty{mycommfont}

\author{
    \IEEEauthorblockN{ Yixuan Wang\IEEEauthorrefmark{2},
    Weichao Zhou\IEEEauthorrefmark{1},
    Jiameng Fan\IEEEauthorrefmark{1}, Zhilu Wang\IEEEauthorrefmark{2}, 
    Jiajun Li\IEEEauthorrefmark{4},\\
    Xin Chen\IEEEauthorrefmark{3}, Chao Huang\IEEEauthorrefmark{4}, Wenchao Li\IEEEauthorrefmark{1}, Qi Zhu\IEEEauthorrefmark{2}}\\
    \IEEEauthorblockA{\IEEEauthorrefmark{1}Department of Electrical and Computer Engineering, Boston University, Boston, MA, USA
    \\\{zwc662, jmfan, wenchao\}@bu.edu}\\
    \IEEEauthorblockA{\IEEEauthorrefmark{2}Department of Electrical and Computer Engineering, Northwestern University, Evanston, IL, USA
    \\\{yixuanwang2024, zhilu.wang, qzhu\}@northwestern.edu}\\
    \IEEEauthorblockA{\IEEEauthorrefmark{3}Department of Computer Science, University of Dayton, Dayton, OH, USA
    \\xchen4@udayton.edu}\\
    \IEEEauthorblockA{\IEEEauthorrefmark{4}Department of Computer Science, University of Liverpool, Liverpool, UK
    \\\{chao.huang2, j.li234\}@liverpool.ac.uk}
    \thanks{Yixuan Wang and Weichao Zhou contributed equally to this work.}
}

\begin{document}


    \title{POLAR-Express: Efficient and Precise \\ Formal Reachability Analysis of \\ Neural-Network Controlled Systems}

\maketitle


\begin{abstract}
Neural networks (NNs) playing the role of controllers have demonstrated impressive empirical performances on challenging control problems. 
However, the potential adoption of NN controllers in real-life applications also gives rise to a growing concern over the safety of these neural-network controlled systems (NNCSs), especially when used in safety-critical applications. 
In this work, we present POLAR-Express, an efficient and precise formal reachability analysis tool for verifying the safety of NNCSs. POLAR-Express uses Taylor model arithmetic to propagate Taylor models (TMs) across a neural network layer-by-layer to compute an overapproximation of the neural-network function. 
It can be applied to analyze any feed-forward neural network with continuous activation functions. 
We also present a novel approach to propagate TMs more efficiently and precisely across ReLU activation functions. 
In addition,
POLAR-Express provides parallel computation support for the layer-by-layer propagation of TMs, thus significantly improving the efficiency and scalability over its earlier prototype POLAR. 
Across the comparison with six other state-of-the-art tools on a diverse set of benchmarks, POLAR-Express achieves the best verification efficiency and tightness in the reachable set analysis.

\end{abstract}

%
%

\section{Introduction} \label{sec:intro}

Neural networks have been successfully used as the central decision-makers in a variety of tasks such as autonomous vehicles~\cite{bojarski2016end, liu2022physics}, aircraft collision avoidance systems~\cite{julian2016policy}, robotics~\cite{levine2016end}, HVAC control~\cite{xu2020one, xu2021learning}, and other autonomous cyber-physical systems (CPS)~\cite{Julian2017}.
Neural-network controllers can be trained using machine learning techniques such as reinforcement learning~\cite{mnih2015human, Lillicrap2016ContinuousCW} and imitation learning from human demonstrations~\cite{abbeel2004apprenticeship,ng2000algorithms} or trajectories generated by model-predictive controllers~\cite{pan2018agile}.
However, the use of neural-network controllers also gives rise to new challenges in verifying the safety of these systems
due to the nonlinear and highly parameterized nature of neural networks and their closed-loop formations with 
dynamical systems
~\cite{huang2019reachnn,Dutta_Others__2019__Reachability,ivanov2018verisig,tran2020nnv}. In this work, we consider the following reachability verification problem of neural-network controlled systems (NNCSs).

\begin{definition}[Reachability Problem of NNCSs]
 The reachability problem of an NNCS is to verify whether a given state is reachable from an initial state of the system, whereas the bounded-time version of this problem is to verify whether a given state is reachable within a given bounded time horizon.
\end{definition}

Uncertainties around the initial state, such as those inherent in state measurement or localization systems, or scenarios where the system can start from anywhere in an initial space, require the consideration of an \textit{initial state set} as opposed to a single initial state in the reachability problem. 
It is worth noting that simulation-based
testings~\cite{zhang2020machine}, which sample initial states from the initial state set,  
cannot provide formal safety guarantees such as ``no system trajectory from the initial state set will lead to an obstacle collision."
In this paper, we consider reachability analysis as the class of techniques that aim at tightly over-approximating the set of all reachable states of the system starting from an initial state set. We illustrate an example of this reachability analysis of a closed-loop system in Figure~\ref{fig:reach}.

\begin{figure}[h!]
    \centering
    \includegraphics[height = 3.5cm]{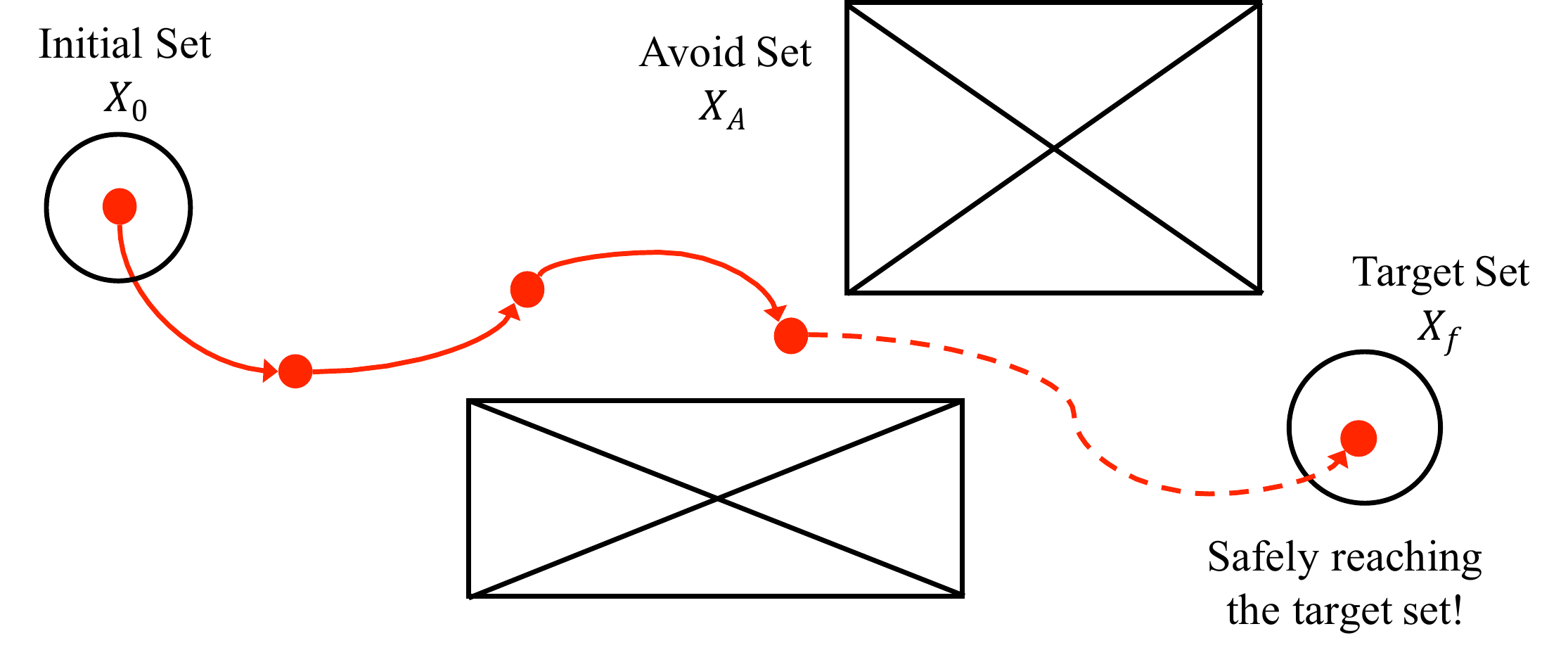}
    \caption{An illustrating example of reachability (for a reach-avoid specification). The system shown in the figure starts from a state in the initial state set $X_0$. Each red dot is the system state at the end of a control time-step, and the red curve is the system trajectory where each solid part indicates the system trajectory between two consecutive red-dot states. In this example, the system safely reaches the target set $X_f$ without hitting the avoid set $X_A$. The reachability verification problem is to check whether this is the case for all initial states in $X_0$ (for some given upper bound on the number of control time-steps). 
    }
    \label{fig:reach}
\end{figure}

Reachability analysis of general NNCSs is notoriously difficult due to nonlinearities that exist in both the neural-network controller and the plant. 
The closed-loop coupling of the neural-network controller with the plant adds another layer of complexity. 
To obtain a tight overapproximation of the reachable sets, reachability analysis needs to
\textit{track state dependencies across the closed-loop system and across multiple time steps}.
While this problem has been well studied in traditional closed-loop systems without neural-network controllers~\cite{Alur+/1995/hybrid_systems,Nedialkov/2011/vnode-lp,Frehse+/2011/SpaceEx,Althoff/2015/CORA,Chen+/2013/flowstar,Kong+/2015/dReach}, it is less clear whether it is important to track the state dependency in NNCSs and how to track the dependency efficiently given the complexity of neural networks. 
This paper aims to bring clarity to these questions by comparing different approaches for solving the NNCS reachability problems. 

Existing reachability analysis techniques for NNCSs typically use reachability analysis methods for dynamical systems as subroutines.  
For general nonlinear dynamical systems,
the problem of \textit{exact reachability is undecidable}~\cite{undecidability_ODE}. 
Thus, methods for reachability analysis of nonlinear dynamical systems aim at computing a \textit{tight over-approximation of the reachable sets}~\cite{dreossi2016parallelotope,lygeros1999controllers,yang2016linear,prajna2004safety,huang2017probabilistic,Frehse+/2011/SpaceEx,Chen+/2013/flowstar,Althoff/2015/CORA}.
On the other hand, 
there is also rich literature on verifying neural networks. 
Most of these verification techniques boil down to the problem of 
estimating or over-approximating the \textit{output ranges} of the network~\cite{huang2017safety,katz2017reluplex,Dutta+Others/2017/Reachability,wang2018formal,singh2018fast,wang2021beta}. 
The existence of these two bodies of work gives rise to a straightforward combination of NN output range analysis with reachability analysis of dynamical systems for solving the NNCS reachability problem. 
However, early works have shown that this naive combination with a non-symbolic interval-arithmetic-based~\cite{Moore+Others/2009/Interval} output range analysis suffers from large over-approximation errors when computing the reachable sets of the overall system~\cite{Dutta_Others__2019__Reachability,huang2019reachnn}.
The primary reason for its poor performance is the lack of  consideration of the interactions between the neural-network controller and the plant dynamics. 
Recent advances in the field of NN verification feature more sophisticated techniques that can yield tighter output range bounds and track the input-output dependency of an NN via symbolic bound propagation~\cite{wang2018formal,zhang2018efficient,wang2021beta}.
This opens up the possibility of improvement for the aforementioned combination strategy 
by substituting the non-symbolic interval-arithmetic-based technique with these new symbolic bound estimation techniques. 

New techniques have also been developed to directly address the verification challenge of NNCSs. Early works mainly use
\textit{direct end-to-end over-approximation} \cite{Dutta_Others__2019__Reachability,huang2019reachnn,fan2020reachnn} of the neural-network function, i.e. computing a function approximation of the neural network with guaranteed error bounds. While this approach can better capture the input-output dependency of a neural network compared to output ranges, it suffers from efficiency problems due to the need to sample from the input space. As a result, 
this type of technique cannot handle systems with more than a few dimensions. 
This approach is superseded by more recent techniques that leverage
\textit{layer-by-layer propagation} in the neural network \cite{ivanov2018verisig,verisig2,ivanov2020verifying,huang2022polar}.
Layer-by-layer propagation techniques have the advantage of being able to exploit the structure of the neural network. 
They are primarily based on 
propagating \textit{Taylor models} (TMs) layer by layer via 
Taylor model arithmetic 
to more efficiently obtain a function over-approximation of the neural network. 

\noindent\textbf{Scope and Contributions.}
We present POLAR-Express, a significantly enhanced version of our earlier prototype tool POLAR~\cite{huang2022polar}. Inherited from POLAR~\cite{huang2022polar}, POLAR-Express uses layer-by-layer propagation of TMs to compute function over-approximations of neural-network controllers. Our technique is applicable to general feed-forward neural networks with continuous (but not necessarily differentiable) activation functions. Compared with POLAR, POLAR-Express has the following new features. 
\begin{itemize}
    \item A more efficient and precise method for propagating TMs across non-differentiable ReLU activation functions.
    \item Multi-threading support to parallelize the computation in the layer-by-layer propagation of Taylor models for neural-network controllers, which significantly improved the efficiency and scalability of our approach for complex systems.
    \item Comprehensive experimental evaluation that includes new comparisons with recent state-of-the-art tools such as RINO~\cite{DBLP:conf/cav/GoubaultP22}, CORA~\cite{DBLP:journals/corr/abs-2210-10691}, and Juliareach~\cite{DBLP:conf/aaai/0001FG22}. Across a diverse set of benchmarks and tools, POLAR-Express achieves state-of-the-art verification efficiency and tightness of over-approximation in the reachable set analysis, outperforming all existing tools.
\end{itemize}

More specifically, compared with the existing literature~\cite{wang2021beta,tran2020nnv,DBLP:journals/corr/abs-2210-10691, DBLP:conf/aaai/0001FG22, DBLP:conf/cav/GoubaultP22, huang2022polar, verisig2}, we provide the most comprehensive experimental evaluation across a wide variety of NNCS benchmarks including neural-network controllers with different activation functions and dynamical systems with up to 12 states. In terms of the approach to over-approximating a neural-network controller, existing tools can be categorized into two classes. The first class shares the common idea of integrating NN output range analysis techniques with reachability analysis tools for dynamical systems, such as $\alpha,\beta$-CROWN~\cite{wang2021beta}, NNV~\cite{tran2020nnv}, CORA~\cite{DBLP:journals/corr/abs-2210-10691}, JuliaReach~\cite{DBLP:conf/aaai/0001FG22}, and RINO~\cite{DBLP:conf/cav/GoubaultP22}. The second class focuses on passing symbolic dependencies across the NNCS and across multiple control steps during reachability analysis, such as POLAR-Express and Verisig 2.0~\cite{verisig2}. Through the comparisons, we hope this paper can also serve as an accessible introduction to these analysis techniques for those who wish to apply them to verify NNCSs in their own applications 
and for those who wish to dive more deeply into the theory of NNCS verification.

\section{Background} \label{sec:related}

We first introduce the technical preliminaries and review existing techniques for the safety verification of NNCSs.
\change{
An NNCS is often defined by an ODE that is governed by a feed-forward neural network at discrete times. Although it is undecidable to know if a state is reachable for an NNCS starting from an initial state, we may compute an over-approximated set of reachable states. The safety of an NNCS can be proven by showing that the reachable set of this NNCS does not contain any unsafe state.
Although more general safety or robustness of an NNCS can be proved by computing an invariant for the system reachable states~\cite{9301422, wang2020energy}, it is still hard to handle a large number of system variables and general nonlinear dynamics. Hence, most of the existing methods for reachability analysis use the \emph{set propagation} scheme~\cite{ChenS22}. That is, an over-approximation of the reachable set in a bounded time horizon can be obtained by iteratively computing a super-set for the reachable set in a time step and propagating it to the set computation for the next step. More precisely, starting from a given initial set $X_0$, a set propagation method computes an over-approximation of the reachable set $\Phi_{t\in [0,\delta]}(X_0,t)$ where $\delta$ is the time step, $\Phi$ denotes the system's evolution function (flowmap) that is often unknown. It then repeats the above work from the obtained reachable set over-approximation at the end of the previous step and computes a new over-approximation for the current one. The over-approximation segments are called \emph{flowpipes}.
Such a scheme has been proven to be effective to handle various system dynamics and efficient in handling large numbers of state variables~\cite{Dutta_Others__2019__Reachability,tran2020nnv,ivanov2020verifying,DBLP:conf/aaai/0001FG22,huang2022polar,DBLP:conf/cav/GoubaultP22}. 

Algorithm~\ref{algo:flowpipe_nncs} shows the main framework of set propagation for NNCSs. Starting with a given initial set $X_0$, the main algorithm repeatedly performs the following two main steps to compute the flowpipes in the $(i+1)$-th control step for $i = 0,1,\dots, K-1$: (a) \textit{Computing the range $\mathcal{U}_i$ of the control input.} This task is to compute the output range of the NN controller w.r.t. the current system state. Since the current system state is a subset of the latest flowpipe, $\mathcal{U}_i$ is computed as an over-approximate set. (b) \textit{Flowpipe construction for the continuous dynamics.} According to the obtained range $\mathcal{U}_i$ of the constant control inputs, the reachable set in the current control step can be obtained using a flowpipe computation method for ODEs.

\begin{algorithm} [t]
 \caption{Reachable set computation for NNCSs based on set-propagation.}\label{algo:flowpipe_nncs}
 \begin{algorithmic}[1]
  \REQUIRE Definition of the system modules, number of control steps $K$, the initial state set $X_0$.
  \ENSURE Over-approximation of the reachable set in $K$ control steps.
  \STATE $\mathcal{R} \leftarrow \emptyset$; \COMMENT{the resulting over-approximation set}
  \STATE $\mathcal{X}_0 \leftarrow X_0$;
  \FOR{$i=0$ to $K-1$}
   \STATE Computing an over-approximation $\mathcal{U}_i$ for the NN output w.r.t. the input set $\mathcal{X}_i$;
   \STATE Computing a set $\mathcal{F}$ of flowpipes for the plant dynamics from the initial set $\mathcal{X}_i$ in a control step;
   \STATE $\mathcal{R} \leftarrow \mathcal{R} \cup \mathcal{F}$;
   \STATE Evaluating an over-approximation for the reachable set at the end of the control step and assigning it to $\mathcal{X}_{i+1}$;
  \ENDFOR
  \RETURN $\mathcal{R}$.
 \end{algorithmic}
\end{algorithm}

The existing methods can be mainly classified into the following two groups based on their over-approximation purposes.

\noindent\textbf{(I) Pure range over-approximations for reachable sets.}
The techniques in this group aim at directly over-approximating the range of the reachable set using geometric or algebraic representations such as intervals~\cite{Jaulin+/2001/applied_interval_analysis}, zonotopes~\cite{Ziegler/1995/Lectures} or other sets represented by constraints. Such an approach can often be developed by designing the over-approximation methods for the plant and controller individually and then using a higher-level algorithm to make the two methods work cooperatively for the closed-loop system. Many existing tools for computing the reachable set over approximations under the continuous dynamics defined by linear or nonlinear ODEs can be used to handle the plant, such as VNODE-LP\cite{Nedialkov/2011/vnode-lp}, SpaceEx~\cite{Frehse+/2011/SpaceEx}, CORA~\cite{Althoff/2015/CORA}, and Flow*~\cite{Chen+/2013/flowstar}. On the other hand, the task of computing the output range of a neural network can be handled by the output range analysis techniques most of which are developed in the recent years~\cite{duttoutputa2018, tjeng2019evaluating, cheng2017maximum, Lomuscio+Maganti/2017/Approach,huang2017safety,katz2017reluplex,ruan2018reachability,gehr2018ai2, singh2018fast,wang2018formal,tran2020nnv,zhang2018efficient,wang2021beta}.
The main advantage of the techniques in this group is twofold. Firstly, there is no need to develop a new technique from scratch, and the correctness of the composed approach can be proved easily based on the correctness of the existing methods for the subtasks. Secondly, the performance of the approach is often good on simple case studies since it can use well-engineered tools as primitives. However, since those methods mainly focus on the pure range over-approximation work, and do not just lightly track the dependencies among the state variables under the system dynamics, it may cause heavy accumulation of over-approximation error when the plant dynamics is nonlinear or the initial set is large, make the resulting bounds less useful in proving properties of interest.

\begin{figure}[t]
    \centering
    \includegraphics[width=.45\textwidth]{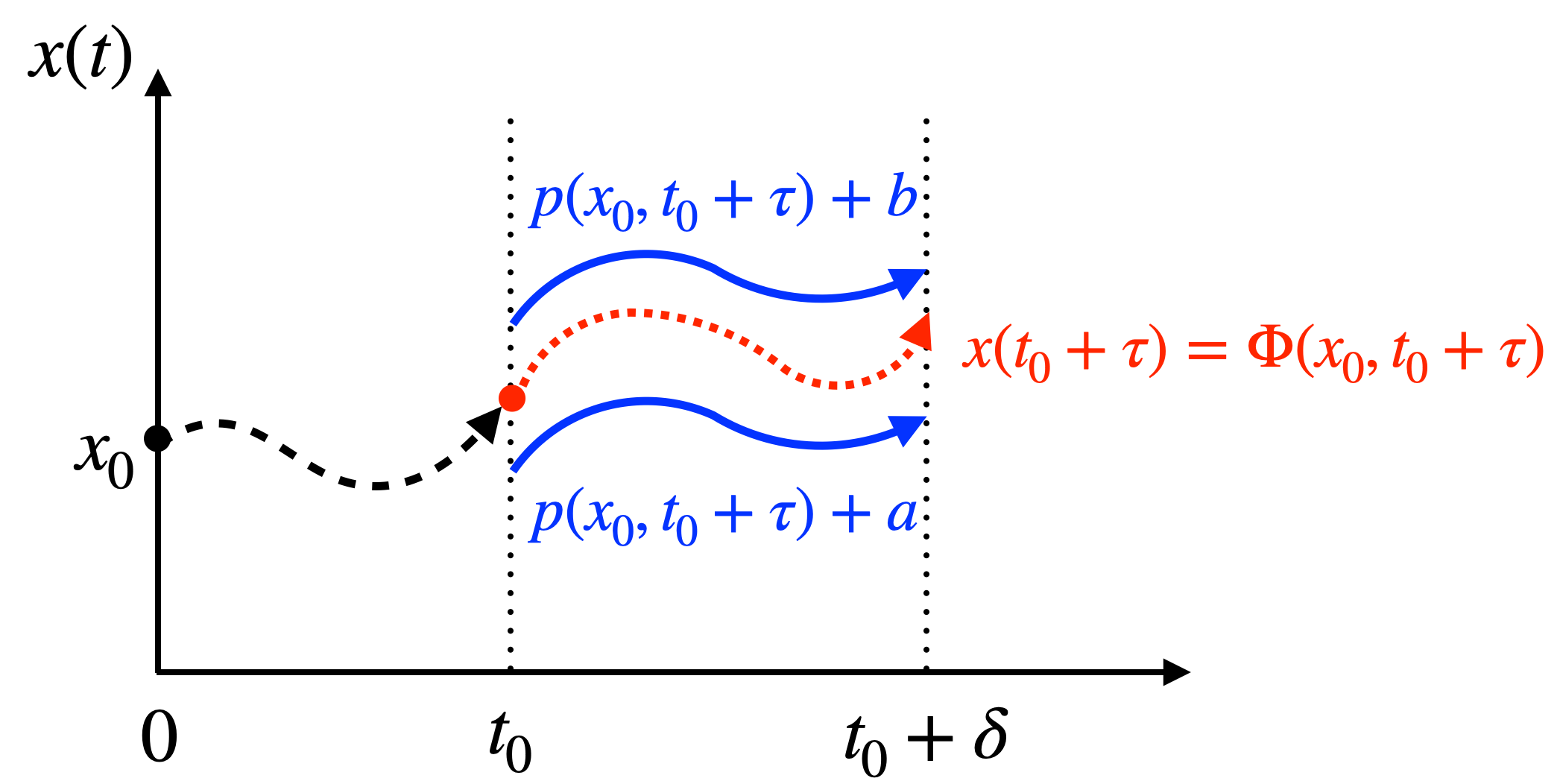}
    \caption{Taylor model over-approximation of a flowmap function.}
    \label{fig:tm_flowpipe}
\end{figure}

\noindent\textbf{(II) Functional over-approximations for system evolution.}
The reachable set over-approximation methods in this category focus on a more challenging task than only over-approximating the reachable set ranges. They seek to compute an over-approximate function for the flowmap $\Phi$ of an NNCS. As we pointed out in the previous section, $\Phi$ is a function only in the variables representing the initial state and the time, and it often does not have a closed-form expression. However, it can be over-approximated by a Taylor model (TM) over a bounded time interval. Figure~\ref{fig:tm_flowpipe} gives an illustration in which the TM $p(\vx_0,t_0+\tau) + [a,b]$ is guaranteed to contain the range of the function $\Phi(\vx_0,t_0 + \tau)$ for any initial state $\vx_0$ and $t\in [0,\delta]$. In practice, we usually require $\vx_0$ to be in a bounded set. Such a TM provides a functional over-approximation rather than a pure range over-approximation which allows tracking the dependency from a reachable state to the initial state approximately.
Functional over-approximations often can handle more challenging reachability analysis tasks, in which larger initial sets, nonlinear dynamics, or longer time horizons are specified.
Recent work has applied interval, polynomial, and TM arithmetic to obtain over-approximations for NNCS evolution~\cite{Dutta_Others__2019__Reachability,huang2019reachnn,IvanovCWAPL21,huang2022polar}. Those techniques are often able to compute more accurate flowpipes than the methods in the other group. On the other hand, the functional over-approximation methods are often computationally expensive due to the computation of nonlinear multivariate polynomials for tracking the dependencies.

\noindent\textbf{Existing tools.}
We consider the following tools in the experimental evaluation: NNV~\cite{tran2020nnv}, Verisig 2.0~\cite{verisig2}, CORA~\cite{DBLP:journals/corr/abs-2210-10691}, JuliaReach~\cite{DBLP:conf/aaai/0001FG22} and RINO~\cite{DBLP:conf/cav/GoubaultP22}. Additionally we also simply combine the use of $\alpha,\beta$-CROWN~\cite{wang2021beta} and Flow*~\cite{Chen+/2013/flowstar} to provide a baseline for the performance of pure range over-approximation. We summarize the key aspects of the tools in Table~\ref{tab:tool_summary}. Basically, NNV, CORA, JuliaReach, and RINO compute range over-approximations, while Verisig 2.0 and POLAR-Express compute functional over-approximations.

\begin{table}[t]
    \centering
    \caption{Summary of the tools evaluated in this paper.}
    \label{tab:tool_summary}
    \begin{adjustbox}{width=.95\columnwidth}
        \begin{tabular}{@{}c|c|c|c|c@{}}
            \toprule
            Tool & Category & \stackanchor{Plant}{Dynamics} & \stackanchor{Activation}{Function} & \stackanchor{Set}{Representation} \\
            \midrule
            \midrule
            \stackanchor{$\alpha,\beta$-CROWN}{+ Flow*} & (I) & nonlinear & continuous & \stackanchor{Interval and}{Taylor model} \\
            \midrule
            NNV & (I) & \stackanchor{discrete linear,}{continuous (CORA)} & \stackanchor{ReLU, tanh,}{sigmoid} & ImageStar \\
            \midrule
            Juliareach & (I) & nonlinear & continuous & Zonotope + Taylor model \\
            \midrule
            CORA & (I) & nonlinear & continuous & Polynomial zonotope \\
            \midrule
            RINO & (I) & nonlinear & differentiable & interval + interval Taylor series \\
            \midrule
            Verisig 2.0 & (II) & nonlinear & differentiable & Taylor model \\
            \midrule
            POLAR-Express & (II) & nonlinear & continuous & Taylor model \\
            \bottomrule
        \end{tabular}
    \end{adjustbox}
\end{table}

\noindent\textbf{Taylor models.}
Taylor models are originally proposed to compute higher-order over-approximations for the ranges of continuous functions (see~\cite{Berz+Makino/1998/Verified}). They can be viewed as a higher-order extension of intervals. A \emph{Taylor model (TM)} is a pair $(p,I)$ wherein $p$ is a polynomial of degree $k$ over a finite group of variables $x_1,\dots,x_n$ ranging in an interval domain $D\subset \reals^n$, and $I$ is the remainder interval. Given a smooth function $f(\vx)$ with $\vx\in D$ for some interval domain $D$, its TM can be obtained as $(p(\vx), I)$ such that $p$ is the Taylor expansion of $f$ at some $\vx_0 \in D$, and $I$ is an interval remainder such that $\forall \vx\in D.(f(\vx) \in p(\vx) + I)$, i.e., $p + I$ is an over-approximation of $f$ at any point in $D$. When the order of $p$ is sufficiently high, the main dependency of the $f$ mapping can be captured in $p$. Basically, the polynomial $p$ can be any polynomial approximation of the function $f$, and it is unnecessary to only use Taylor approximations.

When a function $f(\vx)$ is overly approximated by a TM $(p(\vx), I)$ w.r.t. a bounded domain $D$, the approximation quality, i.e., size of the overestimation, is directly reflected by the width of $I$, since $f(\vx) = p(\vx)$ for all $\vx\in D$ when $I$ is zero by the TM definition.
Given two order $k$ TMs $(p_1(\vx), I_1)$ and $(p_2(\vx), I_2)$ which are over-approximations of the same function $f(\vx)$ w.r.t. a bounded domain $D \subset \reals^n$, we use $(p_1(\vx), I_1) \prec_k (p_2(\vx), I_2)$ to denote that the width of $I_1$ is smaller than the width of $I_2$ in all dimensions, i.e., $(p_1(\vx), I_1)$ is a more accurate over-approximation of $f(\vx)$ than $(p_2(\vx), I_2)$.

TMs have been proven to be powerful over-approximate representations for the flowmap of nonlinear continuous and hybrid systems~\cite{Chen+/2012/taylor_models,Chen/2015/phd,Chen+Sankaranarayanan/2016/Decomposed}.
Although polynomial zonotopes~\cite{polynomial_zonotopes} are also polynomial representations, they are not expressed in the same variables as the system flowmap functions and therefore not functional over-approximations. Interval Taylor Series (ITS) are univariate polynomials in the time variable $t$ where the coefficients are intervals. ITS are often used as nonlinear range over-approximations for ODE solutions~\cite{Nedialkov/2011/vnode-lp}.

}

\section{Problem Formulation} \label{sec:problem}

We use the formal model presented in Figure~\ref{fig:architecture} to describe the behavior of an NNCS. It is a composition of four modules, each of which models the evolution or input-output mapping of the corresponding component in an NNCS. The top three modules form the controller of the system, it retrieves the sensor data $\vy$, computes the control input $\vu$, and applies it to the plant at discrete times $t = 0,\delta_c,\dots,k\delta_c,\dots$ for a control step size $\delta_c >0$. The roles of the modules are described below.

\begin{figure}[b]
    \centering
    \includegraphics[width=.45\textwidth]{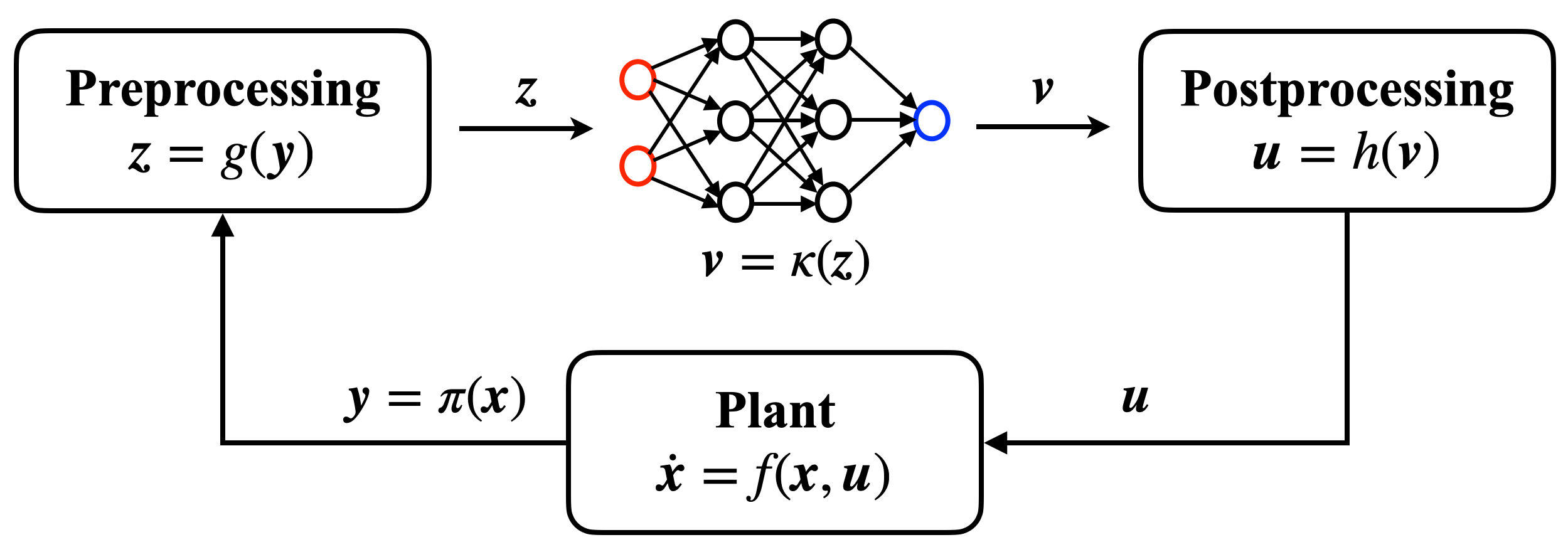}
    \caption{Formal model of NNCS.}
    \label{fig:architecture}
\end{figure}

\noindent\textbf{Plant.} This is a model of the physical process. We use an ODE over the $n$ state variables $\vx$ and $m$ control inputs $\vu$ to model the evolution of the physical process such as the movement of a vehicle, the rotation of a DC motor, or the pitch angle change of an aircraft. In the paper, we collectively represent a set of ordered variables $x_1,\dots,x_n$ by $\vx$. We only consider the ODEs which are at least locally Lipschitz continuous such that its solution w.r.t. an initial condition $\vx(0) = \vx_0\in \reals^n$ is unique~\cite{Perko/2006/ODE}.

\noindent\textbf{Preprocessing Module.} This module transforms the sample data. It serves as the gate of the controller. At the time $t = k\delta_c$, for every $k=0,1,\dots$, it retrieves the sensor data $\vy$ which can be viewed as the image under a mapping from the actual system state $\vx(t)$, and further transform it to an appropriate format $\vz$ for the controller's NN component. A typical preprocessing task in a collision avoidance control could be computing the relative distances of the moving objects.

\noindent\textbf{Neural Network.} This is the core computation module for the control inputs. It maps the input data $\vz$ to the output value $\vv$ according to the layer-by-layer propagation rule defined in it. In the paper, we only consider \emph{feed-forward neural networks}. Since the paper focuses on the formal verification of NNCSs, the neural network is explicitly defined as a part of the NNCS.

\noindent\textbf{Postprocessing Module.} This module transforms the NN output value to the control input. Typically, it is used to keep the final control input in the actual actuating range or to filter out inappropriate input values.

We assume that the preprocessing and postprocessing modules can only be defined by a conjunction of \emph{guarded transitions} each of which is in the form of
\begin{equation}\label{eq:guarded_transition}
 \gamma(\vx)\ \rightarrow\ \vx' = \Pi(\vx)
\end{equation}
such that the guard $\gamma(\vx)$ is a conjunction of inequalities in $\vx$, and $\Pi$ is a transformation from $\vx$. Here, we assume that all guards are disjoint, and allow (\ref{eq:guarded_transition}) to have polynomial arithmetic and the elementary functions $\sin(\cdot)$, $\cos(\cdot)$, $\exp(\cdot)$, $\log(\cdot)$, $\sqrt{\ \cdot\ }$. Then the expressiveness is sufficient to define lookup tables.

\begin{example}[2D Spacecraft Docking]\label{exam:docking}
We consider the docking of a spacecraft in a 2D plane. The benchmark is described in~\cite{ravaioli2022safe}. As shown in Figure~\ref{fig:ex1}, the control goal is to steer the spacecraft to the position at the origin while the velocity should be kept in a safe range. The whole benchmark can be modeled by an NNCS with 5 variables: $\vx = (x,y,v_x,v_y,v_{\text{safe}})^T$ wherein $(x,y)$ denotes the position of the spacecraft, $(v_x,v_y)$ denotes the velocity, and $v_{\text{safe}} = 0.2 + 0.002054\sqrt{x^2 + y^2}$ is a particular variable that indicates a position-dependent safe limit on the speed. The dynamics is defined as in Eq.\ref{docking1} where $f_x$ and $f_y$ constitute the control input $\vu=(f_x, f_y)$ which is obtained by a neural network controller $\kappa$. The input $\vz=(\frac{x}{1000}, \frac{y}{1000}, 2v_x, 2v_y, \sqrt{v_x^2 + v_y^2}, v_{safe})$ of the neural network is prepossessed from $\vx$. 
The output $\vv = (u_x, u_y)$ of the neural network is postprocessed to $f_x=\tanh(u_x), f_y = \tanh(u_y)$. 
\begin{equation}\label{docking1}
    \begin{aligned}
    & \dot{x} = v_x,\quad \dot{y} = v_y, \quad \dot{v}_{safe}=\frac{0.002054 (x\cdot v_x + y\cdot  v_y)}{v_{safe}}\\
    & \dot{v}_x= 0.002054 v_y+3\times (0.001027)^2 x + \frac{f_x}{12},\\  &\dot{v}_y=-0.002054 v_x + \frac{f_y}{12}, \\
    \end{aligned}
\end{equation}
\end{example}

 \begin{figure}[tp!]
\captionsetup[subfloat]{farskip=2pt,captionskip=1pt}
	\centering	
	\subfloat[][]{\scalebox{1.0}{%
		\includegraphics[width=0.45\textwidth]{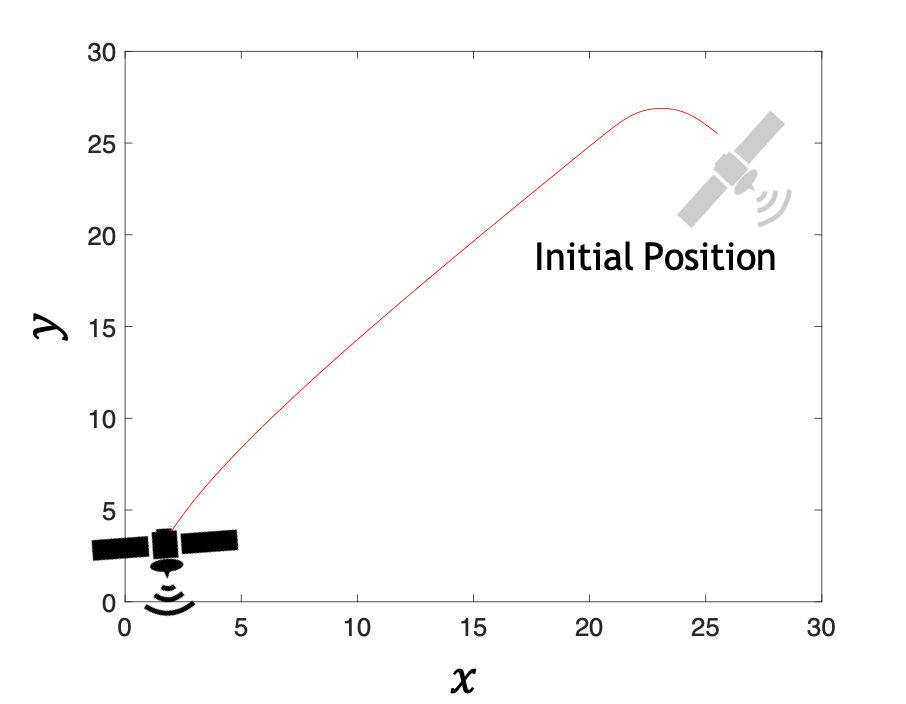}%
		\label{fig:ex1_pos}%
	}}\ 
	\subfloat[][]{\scalebox{1.05}{%
		\includegraphics[width=0.47\textwidth]{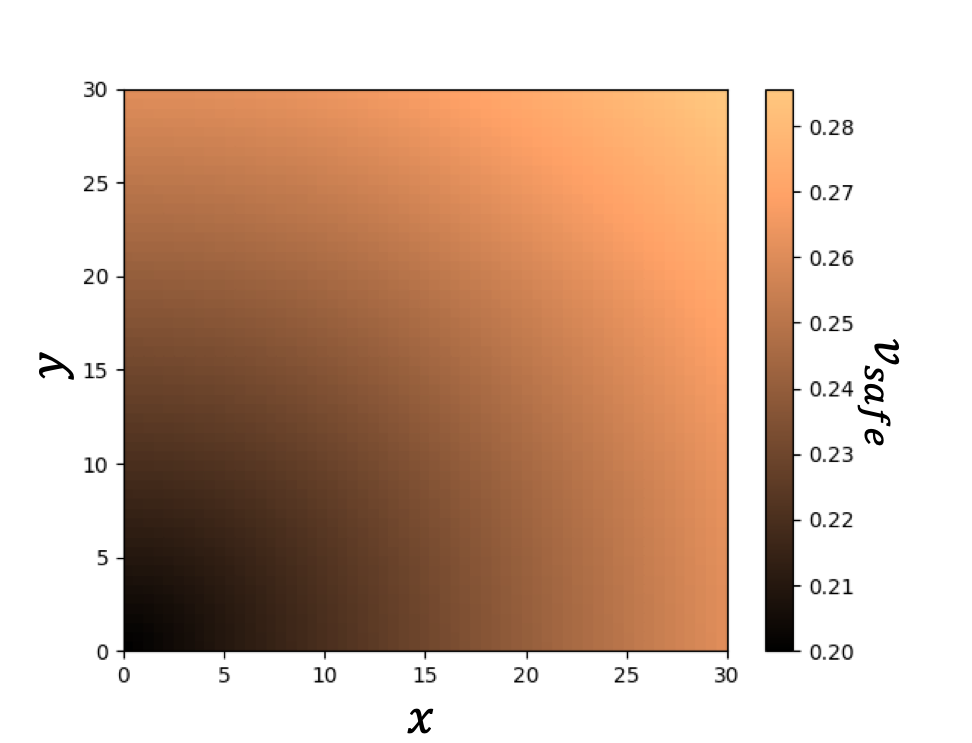}%
		\label{fig:ex1_velsafe}%
	}}
	\caption{The 2D spacecraft docking environment. (a) The goal is to start from some initial position and move towards the origin. (b) The safe speed limit $v_{safe}$ decreases as the position of the spacecraft approaches the origin.}
	\label{fig:ex1}
\end{figure}





\noindent\textbf{Executions of NNCSs.}
Starting from an initial state $\vx_0\in \reals^n$, for all $i=0,1,\dots$, the system state $\vx(t)$ in the $(i+1)$-st control step $t\in [i\delta_c,(i+1)\delta_c]$ is defined by the solution of the ODE $\dot{\vx} = f(\vx,\vu_i)$ w.r.t. the initial state $\vx(i\delta_c)$ and the control input $\vu_i$ which is obtained as $\vu_i = h \circ \kappa \circ g \circ \pi(\vx(i\delta_c))$, where $h, \kappa, g, \pi$ are shown in Fig~\ref{fig:architecture}. If we denote the solution of the ODE w.r.t an initial state $\vx_0$ and a particular control input $\vu'$ by $\vx(t) = \Phi_f(\vx_0,t,\vu')$, the system state at a time $t\in [i\delta_c, (i+1)\delta_c]$ for any $i\geq 0$ from the initial state $\vx_0$ can be expressed recursively by
\begin{equation}\label{eq:evolution}
  \begin{aligned}
   & \vx(t) = \Phi_f(\vx(i\delta_c),t-i\delta_c,\vu_i) \\
   & \text{where } \vu_i = h \circ \kappa \circ g \circ \pi(\vx(i\delta_c))
  \end{aligned}
\end{equation}
such that $\vx(0) = \vx_0$. We also call this state a \emph{reachable state}. Without noises or disturbances from the environment, an NNCS has deterministic behavior, and its evolution can be defined by a \emph{flowmap} function in the form of $\Phi(\vx_0,t)$, i.e., the reachable state from an initial state $\vx_0$ at a time $t$ is $\vx(t) = \Phi(\vx_0,t)$, and it is \emph{uniquely determined by the initial state and the time}. Unfortunately, $\Phi$ usually does not have a closed-form expression.

The reachability analysis task with respect to two given system states $\vx,\vx'$ and an NNCS asks whether $\vx'$ is reachable from $\vx$ under the evolution of the system.
Reachability analysis plays a key role in the safety verification of dynamical systems. However, it is a notoriously difficult task due to the undecidability of the reachability problem on the systems defined by even nonlinear difference equations~\cite{10.1007/BFb0109720}. In order to prove the safety of a system, most of the reachability techniques seek to compute an over-approximation of the reachable set. If no unsafe state is contained in the over-approximated reachable set, then the system is guaranteed to be safe.


\section{Framework of POLAR-Express}\label{sec:polar_express}

We present the POLAR-Express framework in Algorithm~\ref{algo:polar_express} to compute flowpipes for NNCSs. It uses the standard set-propagation framework Algorithm~\ref{algo:flowpipe_nncs} but has the following novel elements: \textit{(1) Polynomial over-approximation for activation functions using \emph{B{\'e}zier curves}.}
\textit{(2) Symbolic representation of TM remainders in layer-by-layer propagation.}
\textit{(3) A seamless integration of the above techniques to compute accurate flowpipes for NNCSs.}  \textit{(4) A more precise and efficient method for propagating TMs across ReLU activation function.} \textit{(5) Multi-threading support to parallelize the computation in the layer-by-layer propagation of TMs for NN.} The details are explained below.

\begin{algorithm} [t]
 \caption{Framework of POLAR-Express.}\label{algo:polar_express}
 \begin{algorithmic}[1]
  \REQUIRE Plant dynamics $\vx' = f(\vx,\vu)$, preprocessing $g(\cdot)$, postprocessing $h(\cdot)$, NN controller $\kappa(\cdot)$, number of control steps $K$, initial set $X_0$.
  \ENSURE Over-approximation of the reachable set over the time interval of $[0, K\delta_c]$ where $\delta_c$ is the control step size.
  \STATE $\mathcal{R} \leftarrow \emptyset$;
  \STATE $\mathcal{X}_0 \leftarrow X_0$;
  \FOR{$i=0$ to $K-1$}
   \STATE Computing a superset $\mathcal{Y}_i$ for the range of $\pi(\mathcal{X}_i)$;
   \STATE Computing a superset $\mathcal{Z}_i$ for the range of $g(\mathcal{Y}_i)$;
   \STATE Computing a superset $\mathcal{V}_i$ for the range of $\kappa(\mathcal{Z}_i)$, with multi-threading support;
   \STATE Computing a superset $\mathcal{U}_i$ for the range of $h(\mathcal{V}_i)$;
   \STATE Computing a set $\mathcal{F}$ of flowpipes for the continuous dynamics $\dot{\vx} = f(\vx,\vu)$ with $\vu\in \mathcal{U}_i$ from the initial set $\vx(0) \in \mathcal{X}_i$ over the time interval of $[i\delta_c,(i+1)\delta_c]$;
   \STATE $\mathcal{R} \leftarrow \mathcal{R} \cup \mathcal{F}$;
   \STATE Evaluating an overapproximation for the reachable set at the time $t = (i+1)\delta_c$ based on $\mathcal{F}$ and assigning it to $\mathcal{X}_{i+1}$;
  \ENDFOR
  \RETURN $\mathcal{R}$.
 \end{algorithmic}
\end{algorithm}

\noindent\textbf{Computing $\mathcal{Y}_i$ and $\mathcal{U}_i$.}
Given a preprocessing or postprocessing module and its input set which is represented as a TM, an output TM can be obtained by computing the reachable sets of the guarded transitions. Given a guarded transition of the form (\ref{eq:guarded_transition}) along with a TM $S$ for $\vx$'s range, the reachable set $S'$, i.e., the range of $\vx'$, can be computed by first computing the intersection $S_I = S\,\cap\,\{\vx\,|\,\gamma(\vx)\}$ and then evaluating $\Pi(S_I)$ using TM arithmetic~\cite{Makino+Berz/2003/Taylor}. Although TMs are not closed under an intersection with a semi-algebraic set, we may use the domain contraction method proposed in~\cite{Chen+/2012/taylor_models} to derive an over-approximate TM for the intersection.

\subsection{Layer-by-Layer Propagation using TMs}

POLAR-Express uses the layer-by-layer propagation scheme to compute a TM output for the NN, and features the following key novelties:
(a) A method to selectively compute Taylor or Bernstein polynomials for activation functions. The purpose is to \emph{derive a smaller error according to the approximated function and its domain}. The Bernstein polynomials are represented in their \emph{B{\'e}zier forms}.
(b) A technique to symbolically represent the intermediate linear transformations of TM interval remainders during the layer-by-layer propagation. The purpose of using Symbolic Remainders \textbf{(SR)} is to \emph{reduce the accumulation of overestimation produced in composing a sequence of TMs}. The approach is described as follows.

\noindent\textbf{Layer-by-Layer Propagation with Multi-Threading Support}.
The framework of layer-by-layer propagation has been widely used to compute NN output ranges. Most of the existing methods use range over-approximations such as intervals with constant bounds~\cite{duttoutputa2018,weng2018evaluating} or linear polynomial bounds~\cite{wang2018formal}, zonotopes~\cite{gehr2018ai2,singh2018fast}, star sets~\cite{tran2020nnv}. Some TM-based methods are also proposed~\cite{Dutta_Others__2019__Reachability,huang2019reachnn,ivanov2020verifying} to obtain functional over-approximations for the input-output mapping of an NN. However, the use of functional over-approximations in the reachability analysis of an NNCS as a whole has not been well investigated. Hence, we propose the following approach to better over-approximate the input-output dependency than the existing state-of-the-art.

\begin{algorithm} [t]
 \caption{Layer-by-layer propagation using polynomial arithmetic and TMs}\label{algo:nn_output}
 \begin{algorithmic}[1]
  \REQUIRE Input TM $(p_1(\vx_0),I_1)$ with $\vx_0\in X_0$, the $M+1$ matrices $W_1,\dots,W_{M+1}$ of the weights on the incoming edges of the hidden and the output layers, the $M+1$ vectors $B_1,\dots,B_{M+1}$ of the neurons' bias in the hidden and the output layers.
  \ENSURE a TM $(p_r(\vx_0),I_r)$ that contains the set $\kappa((p_1(\vx_0),I_1))$.
  \STATE $(p_r,I_r) \leftarrow (p_1,I_1)$;
  \FOR{$i=1$ to $M+1$}
   \STATE $(p_t,I_t) \ \leftarrow\ W_i \cdot (p_r,I_r) + B_i$;
   \STATE Computing a polynomial approximation $p_{\sigma,i}$ for the vector of the current layer's activation functions $\sigma$ w.r.t. the domain $(p_t,I_t)$;
   \STATE Evaluating a conservative remainder $I_{\sigma,i}$ for $p_{\sigma,i}$ w.r.t. the domain $(p_t,I_t)$;
   \STATE $(p_r,I_r) \ \leftarrow\ p_{\sigma,i}(p_t + I_t) + I_{\sigma,i}$;
  \ENDFOR
  \RETURN $(p_r,I_r)$.
 \end{algorithmic}
\end{algorithm}

Algorithm~\ref{algo:nn_output} presents the main framework of our approach without using SR and focuses on the novelty coming from the tighter TM over-approximation for the activation functions (Lines 4 and 5).
Before introducing our selective over-approximation method, we describe how a TM output is computed from a given TM input for a single layer. The idea is illustrated in Fig.~\ref{fig:single_layer_composition}. The circles in the right column denote the neurons in the current layer which is the $(i+1)$-th layer, and those in the left column denote the neurons in the previous layer. The weights on the incoming edges to the current layer are organized as a matrix $W_i$, while we use $B_i$ to denote the vector organization of the biases in the current layer. Given that the output range of the neurons in the previous layer is represented as a TM (vector) $(p_i(\vx_0), I_i)$ where $\vx_0$ are the variables ranging in the NNCS initial set. Then, the output TM $(p_{i+1}(\vx_0), I_{i+1})$ of the current layer can be obtained as follows. First, we compute the polynomial approximations $p_{\sigma_1,i},\dots,p_{\sigma_l,i}$ for the activation functions $\sigma_1,\dots,\sigma_l$ of the neurons in the current layer. Second, interval remainders $I_{\sigma_1,i},\dots,I_{\sigma_l,i}$ are evaluated for those polynomials to ensure that for each $j=1,\dots,l$, $(p_{\sigma_j,i},I_{\sigma_j,i})$ is a TM of the activation function $\sigma_j$ w.r.t. $z_j$ ranging in the $j$-th dimension of the set $W_i(p_i(\vx_0) + I_i)$. Third, $(p_{i+1}(\vx_0,I_{i+1}))$ is computed as the TM composition $p_{\sigma,i}(W_i(p_i(\vx_0) + I_i) + I_{\sigma,i}$ where $p_{\sigma,i}(\vz) = (p_{\sigma_1,i}(z_1),\dots,p_{\sigma_l,i}(z_k))^T$ and $I_{\sigma,i} = (I_{\sigma_1,i},\dots,I_{\sigma_l,i})^T$. Hence, when there are multiple layers, starting from the first layer, the output TM of a layer is treated as the input TM of the next layer, and the final output TM is computed by composing TMs layer-by-layer. \zhou{Besides, we use $(p_{j,i}, I_{j, i})$ for $j=1,\ldots, l$ to represent the TMs associated with the $l$ neurons in a linear layer. The computation of those TMs can be conducted in parallel. So is the propagation through the activation functions in a layer. POLAR-Express realizes such parallelism via multi-threading to retain time efficiency when the dimension of the NN layers is large.}

\begin{figure}
  \begin{center}
    \includegraphics[width=0.4\textwidth]{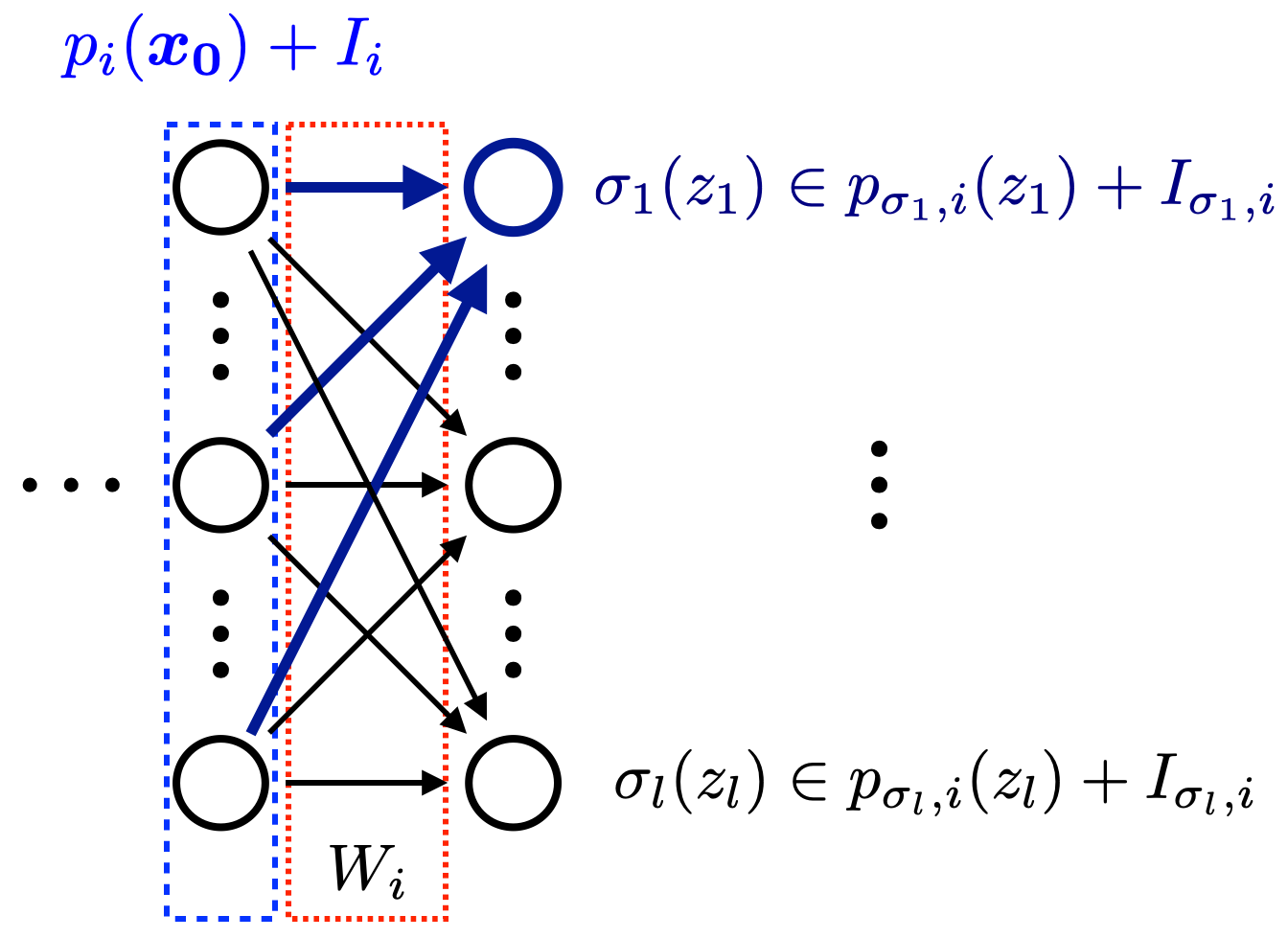}
  \end{center}
  \caption{Single layer propagation}\label{fig:single_layer_composition}
\vspace{-5mm}
\end{figure}

\noindent\textbf{Polynomial Approximations to TMs.}
Basically, a TM only defines an over-approximate mapping and is independent of the approximation method used for the polynomial part. Thus, we consider using both Taylor and Bernstein approximations when propagating through an activation function and choose the one that produces less overestimation after a TM combination. The following example shows that the selection cannot be determined only based on the approximation error. Given the TMs $(p_1,I_1)$, $(p_2,I_2)$ which are both TM over-approximations for the sigmoid function $f(x) = \frac{1}{1 + e^{-x}}$ w.r.t. a TM domain $x\in q(y) + J$:
\[
 \begin{aligned}
  (p_1,I_1) & = (0.5 + 0.25 x - 0.02083 x^3, [\text{-7.93e-5, 1.92e-4}]) \\
  (p_2,I_2) & = (0.5 + 0.24855 x - 0.004583 x^3, [\text{-2.42e-4, 2.42e-4}]) \\
  (q, J) & = (0.1 y - 0.1 y^2, [\text{-0.1,0.1}]).
 \end{aligned}
\]
where $y\in [-1,1]$. However the composition $(p_1(q(y) + J) + I_1)$ produces a TM with the remainder $[-0.0466 , 0.0477]$, while the remainder produces by $p_2(q(y) + J) + I_2$ is $[-0.0253 , 0.0253]$ which is smaller. In other words, a smaller polynomial approximation error does not always lead to a smaller error in combination. Therefore, it motivates us to do a selection after the combination. We generalize this phenomenon by defining the \emph{accuracy preservation problem}, and obviously, the answer is no if TM arithmetic is used.

\begin{definition}[Accuracy preservation problem]
 If both $(p_1(\vx), I_1)$ and $(p_2(\vx), I_2)$ are over-approximations of $f(\vx)$ with $\vx\in D$, and $(p_1(\vx), I_1) \prec_k (p_2(\vx), I_2)$. Given another function $g(\vy)$ which is already over-approximated by a TM $(q(\vy), J)$ whose range is contained in $D$. Then, \emph{does $p_1(q(\vy) + J) + I_1\, \prec_k\, p_2(q(\vy) + J) + I_2$ still hold using order $k$ TM arithmetic?}
\end{definition}

\subsection{Bernstein Over-approximation for Activation Functions}

Now we turn to our Bernstein over-approximation method for activation functions. It first computes a Bernstein polynomial for the function and then evaluates a remainder interval to ensure the over-approximation. The polynomials are in the B{\'e}zier form.

\begin{definition}[Bernstein polynomial]
 Given a continuous function $f(x)$ with $x\in [a,b]$, its order $k$ Bernstein polynomial $p_k(x)$ is defined by
 \begin{equation}\label{eq:berstein}
  \sum_{i=0}^n \left( f\left(a + \frac{i}{k}(b-a)\right) \binom{k}{i} \left(\frac{x-a}{b-a} \right)^i \left(\frac{b-x}{b-a}\right)^{k-i} \right)
 \end{equation}
\end{definition}

\noindent\textbf{Bernstein approximation in B{\'e}zier form.}
The use of Bernstein approximation only requires the activation function to be continuous in $(p_t,I_t)$, and can be used not only in more general situations but also to obtain better polynomial approximations than Taylor expansions (see~\cite{Lorentz/Bernstein}). We first give a general method to obtain a Bernstein over-approximation for an arbitrary continuous function, and then present \textit{a more accurate approach only for ReLU functions}. Given that the activation functions in a layer are collectively represented as a vector $\sigma(\vz)$ and $\vz$ ranges in a TM $(p_t,I_t)$. Then the order $k$ Bernstein polynomial $p_{\sigma_j,i}(z_j)$ for the activation function $\sigma_j$ of the $j$-th neuron. It can be computed as (\ref{eq:berstein}) while $f$ is $\sigma_j$, $a,b$ are the lower and upper bound respectively of the range in the $j$-th dimension of $(p_t,I_t)$, and they can be obtained from an interval evaluation of the TM.

\noindent\textbf{Remainder Evaluation.}
The remainder $I_{\sigma_j,i}$ for the polynomial $p_{\sigma_j,i}$ can be obtain as a symmetric interval $[-\epsilon_j,\epsilon_j]$ such that $\epsilon_j$ is
\[
 \begin{aligned}
    \max_{s=1,{\cdots},m}& \left( \left| p_{\sigma_j,i}(\frac{\overline{Z
    }_j{-}\underline{Z}_j}{m}(s{-}\frac{1}{2}) {+}\underline{Z}_j) \right.\right. \\
    & \left. \left. {-} {\sigma_j}(\frac{\overline{Z}_j-\underline{Z}_j}{m}(s-\frac{1}{2}){+}\underline{Z}_j)\right| {+} L_j{\cdot}\frac{\overline{Z}_j{-}\underline{Z}_j}{m} \right)
 \end{aligned}
\]
wherein $L_j$ is a Lipschitz constant of $\sigma_j$ with the domain $(p_t,I_t)$, and $m$ is the number of samples that are uniformly selected to estimate the remainder. The soundness of the error bound estimation above has been proved in \cite{huang2019reachnn} for multivariate Bernstein polynomials. Since the univariate Bernstein polynomial, which we use in this paper, is a special case of multivariate Bernstein polynomials, our approach is also sound.

\noindent\textbf{A More Precise and Efficient Bernstein Over-approximation for ReLU.}
The above Bernstein over-approximation method works on all continuous activation functions, however, if a function is convex or concave on the domain of interest, a more accurate Bernstein over-approximation that is represented as a TM can be obtained as follows. Given a continuous function $f(x)$ with $x\in D$ such that $f$ is convex on the domain, then the Bernstein polynomials of $f$ are no smaller than $f$ at any point in $D$. Thus, a tight upper bound for $f$ can be computed as one of its Bernstein polynomial $p$, while a tight lower bound can be obtained by moving $p$ straight down by the distance which is equivalent to the maximum difference of $f$ and $p$ for $x\in D$. When $f$ is ReLU and $0\in D$, it is convex on $D$, and its maximum difference to any Bernstein polynomial $p$ is $p(0)$. We give the particular Bernstein over-approximation method for ReLU functions by Algorithm~\ref{algo:bern_relu}. An example is illustrated in Fig.~\ref{fig:bp_relu_ub_lb}.

\begin{lemma}\label{lem:bernstein_bound}
 Given that $p_k(x)$ is the order $k\geq 1$ Bernstein polynomial of a convex function $f(x)$ with $x\in [a,b]$. For all $x\in [a,b]$, we have that (i) $f(x) \leq p_k(x)$ and (ii) $p_{k+1}(x) \leq p_k(x)$.
\end{lemma}
\begin{proof}
 The Lemma is proved in~\cite{goodman/phillips_1999} for the domain $x \in [0,1]$. However, it also holds on an arbitrary domain $x \in [a,b]$ after we replace the lower and upper bounds in the Bernstein polynomials by $a$ and $b$.
\end{proof}

\begin{corollary}
 If $p(x)$ is the order $k\geq 1$ Bernstein polynomial of $\relu(x)$ with $x\in [a,b]$, then $0\leq \relu(x) \leq p(x)$ for all $x\in [a,b]$.
\end{corollary}

\begin{lemma}
 Given that $p(x)$ is the order $k\geq 1$ Bernstein polynomial of $\relu(x)$ with $x\in [a,b]$ such that $a < 0 < b$, then we have that $p(x) - \relu(x)\leq p(0)$ for all $x\in [a,b]$.
\end{lemma}
\begin{proof}
 Since $\relu(x)$ is convex over the domain, by~\cite{goodman/phillips_1999}, so is $p(x)$. Therefore, the second derivative of $p$ w.r.t. $x$ is non-negative. By evaluating the first derivatives of $p$ at $x = a$ and $x = b$, we have that $\frac{dp}{dx}|_{x = a} \geq 0$ and $\frac{dp}{dx}|_{x = b} \leq 1$. Since the first derivatives of $\relu(x)$ are $0$ and $1$ when $x\in [a,0)$ and $x\in (0,b]$ respectively, $p(a) = \relu(a)$, and $p(b) = \relu(b)$, we have that the function $p(x) - \relu(x)$ monotonically increasing when $x\in [a,0]$ and decreasing when $x\in [0,b]$, hence its maximum value is given by $p(0)$.
\end{proof}

\begin{example}
 \zhou{In~Fig\ref{fig:express_vs_ori}, an NNCS with ReLU as the activation functions starts from an initial set near $(0.4, 0.45)$ and moves towards a target set enclosed by the yellow rectangle (see details in benchmark 5~\cite{Dutta_Others__2019__Reachability,huang2019reachnn}). POLAR-Express with this new precise and efficient over-approximation approach for the ReLU function generates tighter flowpipes than POLAR. The runtime of POLAR-Express is $1.8$s while the runtime of POLAR is $7.1$s. This result serves as the evidence that this novel Bernstein Over-approximation for ReLU achieves better verification efficiency and accuracy.}
\end{example}
\begin{figure}[!ht]
    \centering
    \includegraphics[width=0.8\linewidth]{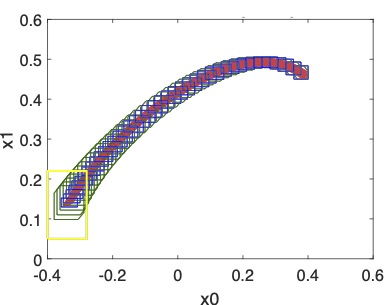}
    \caption{POLAR-Express (blue) generates tighter over-approximations than POLAR (green) under the same hyper-parameters; the red curves are the simulated traces; the yellow rectangle is the target set.} 
    \label{fig:express_vs_ori}
\end{figure}

\begin{algorithm} [t]
 \caption{Efficient and Tight Bernstein over-approximation for ReLU functions}\label{algo:bern_relu}
 \begin{algorithmic}[1]
  \REQUIRE Domain $D = [a,b]$ ($a\leq b$) of the ReLU function, approximation order $k$.
  \ENSURE a TM that overapproximates the ReLU function on $D$.
  \IF{$a \geq 0$}
   \RETURN $(x, [0,0])$;
  \ELSIF{$b\leq 0$}
   \RETURN $(0, [0,0])$;
  \ELSE
   \STATE Computing the order $k$ Bernstein polynomial $p$;
   \STATE $\varepsilon \leftarrow p(0)$;
   \RETURN $(p - 0.5\varepsilon, [-0.5\varepsilon, 0.5\varepsilon])$;
  \ENDIF
 \end{algorithmic}
\end{algorithm}

\begin{figure}
  \begin{center}
    \includegraphics[width=0.48\textwidth]{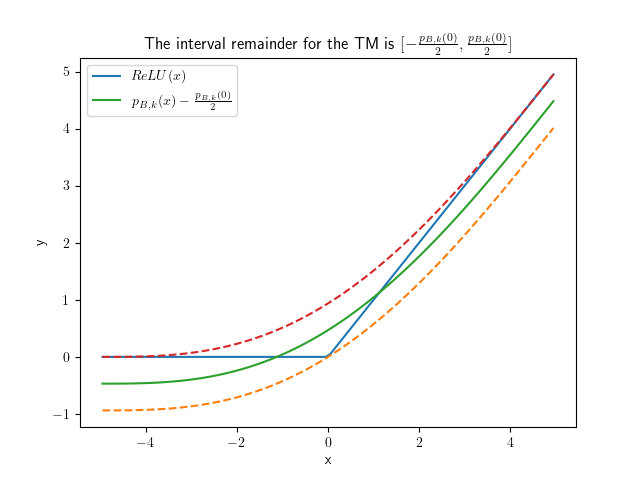}
  \end{center}
  \caption{
  The Taylor model (TM) over-approximation $p(x)+I$ of $\text{ReLU}(x)$ is given by $p(x) = p_{B,k}(x) - \frac{p_{B,k}(0)}{2}$ and $I = [-\frac{p_{B,k}(0)}{2}, \frac{p_{B,k}(0)}{2}])$ where $p_{B,k}(0)$ is the Bernstein polynomial $p_{B,k}(x)$ evaluated at $x=0$. It can be shown that for $x \in [a, b]$ with $a < 0 < b$, the bounds of the interval remainder $I$ are tight for any order $k$ Bernstein polynomials approximation with $k \geq 1$.
  }
  \label{fig:bp_relu_ub_lb}
\end{figure}

\subsection{Using Symbolic Remainders}
A main source of overestimation in interval arithmetic is the computation of linear mappings. Given a box (Cartesian product of single intervals) $I$, its image under a linear mapping $\vx \mapsto A \vx$ is often not a box but has to be over-approximated by a box in interval arithmetic. For a sequence of linear mappings, the resulting box is often unnecessarily large due to the overestimation accumulated in each mapping. It is also known as the \emph{wrapping effect}~\cite{Jaulin+/2001/applied_interval_analysis}. To avoid this class of overestimation, we may symbolically represent the intermediate boxes, and only do an interval evaluation at last. For example, if we need to compute the image of the box $I$ through the linear mappings: $\vx \mapsto A_1 \vx, \dots, \vx \mapsto A_m \vx$, the box $I$ is kept symbolically and the composite mapping is computed as $A' = A_m\cdots A_1$, a tight interval enclosure for the image can be obtained from evaluating $A'I$ using interval arithmetic.

Although TM arithmetic uses polynomials to symbolically represent the variable dependencies, it is not free from wrapping effects since the remainders are always computed using interval arithmetic. Consider the TM composition for computing the output TM of a single layer in  Fig.~\ref{fig:single_layer_composition},
the output TM $p_{\sigma,i}(W_i (p_i(\vx_0) + I_i) + B_i) + I_{\sigma,i}$ equals to $Q_i W_i p_i(\vx_0) + Q_i W_i I_i + Q_i B_i + p_{\sigma,i}^R(W_i (p_i(\vx_0) + I_i) + B_i) + I_{\sigma,i}$ such that $Q_i$ is the matrix of the linear coefficients in $p_{\sigma,i}$, and $p_{\sigma,i}^R$ consists of the terms in $p_{\sigma,i}$ of the degrees $\neq 1$. Therefore, the remainder $I_i$ in the second term can be kept symbolically such that we do not compute $Q_i W_i I_i$ out as an interval but keep its transformation matrix $Q_i W_i$ to the subsequent layers. Given the image $S$ of an interval under a linear mapping, we use $\underline{S}$ to denote that it is kept symbolically, i.e., we keep the interval along with the transformation matrix, and $\overline{S}$ to denote that the image is evaluated as an interval.

Next, we present the use of SR in layer-by-layer propagation. Starting from the NN input TM $(p_1(\vx_0),I_1)$, the output TM of the first layer is computed as
\[
 \begin{aligned}
  & \underbrace{Q_1 W_1 p_1(\vx_0) + Q_1 B_1 + p_{\sigma,1}^R(W_1 (p_1(\vx_0) + I_1) + B_1) + I_{\sigma,1}}_{q_1(\vx_0) + J_1} \\
  & + \underline{Q_1 W_1 I_1}
 \end{aligned}
\]
which can be kept in the form of $q_1(\vx_0) + J_1 + \underline{Q_1 W_1 I_1}$. Using it as the input TM of the second layer, we have the following TM
\[
 \scriptsize
 \begin{aligned}
   & p_{\sigma,2}(W_2 (q_1(\vx_0) + J_1 + \underline{Q_1 W_1 I_1}) + B_2) + I_{\sigma,2} \\
   = & \underbrace{Q_2 W_2 q_1(\vx_0) + Q_2 B_2 + p_{\sigma,2}^R(W_2 (q_1(\vx_0) + J_1 + \overline{Q_1 W_1 I_1}) + B_2) + I_{\sigma,2}}_{q_2(\vx_0) + J_2} \\
   & + \underline{Q_2 W_2 J_1} + \underline{Q_2 W_2 Q_1 W_1 I_1}
 \end{aligned}
\]
for the output range of the second layer. Therefore the output TM of the $i$-th layer can be obtained as $q_i(\vx_0) + \mathbb{J}_i + \underline{Q_iW_i\cdots Q_1W_1 I_1}$ such that $\mathbb{J}_i = J_i + \underline{Q_iW_iJ_{i-1}} + \underline{Q_iW_iQ_{i-1}W_{i-1}J_{i-2}} + \cdots + \underline{Q_iW_i\cdots Q_2W_2 J_1}$.

We present the SR method in Algorithm~\ref{algo:sym_rem} where we use two lists: $\mathcal{Q}[j]$ for $Q_iW_i\cdot\cdots\cdot Q_jW_j$ and $\mathcal{J}[j]$ for $\mathbb{J}_j$ to keep the intervals and their linear transformations. The symbolic remainder representation is replaced by its interval enclosure $I_r$ at the end of the algorithm.

\noindent\textbf{Time and space complexity.}
Although Algorithm~\ref{algo:sym_rem} produces TMs with tighter remainders than Algorithm~\ref{algo:nn_output}, because of the symbolic interval representations under linear mappings, it requires (1) two extra arrays to keep the intermediate matrices and remainder intervals, (2) two extra inner loops which perform $i-1$ and $i-2$ iterations in the $i$-th outer iteration. The size of $Q_iW_i\cdot\cdots\cdot Q_jW_j$ is determined by the rows in $Q_i$ and the columns in $W_j$, and hence the maximum number of neurons in a layer determines the maximum size of the matrices in $\mathcal{Q}$. Similarly, the maximum dimension of $J_i$ is also bounded by the maximum number of neurons in a layer. Because of the two inner loops, time complexity of Algorithm~\ref{algo:sym_rem} is quadratic in $M$, whereas Algorithm~\ref{algo:nn_output} is linear in $M$.

\begin{algorithm}[t]
 \caption{TM output computation using symbolic remainders, input and output are the same as those in Algorithm~\ref{algo:nn_output}}\label{algo:sym_rem}
 \begin{algorithmic}[1]
  \STATE Setting $\mathcal{Q}$ as an empty array which can keep $M+1$ matrices;
  \STATE Setting $\mathcal{J}$ as an empty array which can keep $M+1$ multidimensional intervals;
  \STATE $\mathbb{J} \leftarrow 0$;
  \FOR{$i=1$ to $M+1$}
   \STATE Computing the composite function $p_{\sigma,i}$ and the remainder interval $I_{\sigma,i}$ using the BP technique;
   \STATE Evaluating $q_i(\vx_0) + J_i$ based on $\mathbb{J}$ and $\mathcal{Q}[1]I_1$; 
   \STATE $\mathbb{J} \leftarrow J_i$;
   \STATE $\Phi_i = Q_i W_i$;
   \FOR{$j=1$ to $i-1$}
    \STATE $ \mathcal{Q}[j] \leftarrow \Phi_i \cdot \mathcal{Q}[j]$;
   \ENDFOR
   \STATE Adding $\Phi_i$ to $\mathcal{Q}$ as the last element;
   \FOR{$j=2$ to $i$}
    \STATE $\mathbb{J} \leftarrow \mathbb{J} + \mathcal{Q}[j] \cdot \mathcal{J}[j-1]$;
   \ENDFOR
   \STATE Adding $J_i$ to $\mathcal{J}$ as the last element;
  \ENDFOR
  \STATE Computing an interval enclosure $I_r$ for $\mathbb{J} + \mathcal{Q}[1]I_1$; 
  \RETURN $q_{M+1}(\vx_0) + I_r$.
 \end{algorithmic}
\end{algorithm}

\begin{theorem}\label{thm:main}
 In Algorithm~\ref{algo:polar_express}, if $(p(\vx_0,\tau), I)$ is the $i$-th TM flowpipe computed in the $j$-st control step, then for any initial state $\vc\in X_0$, the box $p(\vc,\tau) + I = p'(\tau) + I$ contains the actual reachable state $\varphi_\mathcal{N}(\vc,(j-1)\delta_c + (i-1)\delta + \tau)$ for all $\tau\in [0,\delta]$.
\end{theorem}

\section{Benchmark Evaluations}\label{sec:benchmarks}



We conduct a comprehensive comparison with state-of-the-art tools across a diverse set of benchmarks. In addition, we discuss in detail the applicability and comparative advantages of different techniques. The experiments were performed on a machine with a 6-core 2.20 GHz Intel i7 CPU and 16GB of RAM.  For tools that can leverage GPU acceleration such as $\alpha,\beta$-CROWN, the experiments were run with the aid of an Nvidia GeForce GTX 1050Ti GPU. \zhou{The multi-threading support is realized by using the C++ Standard Library. Considering the overhead introduced by multi-threading, we also measure the performance of the application under single thread to identify the bottleneck caused by multi-threading.}

\noindent\textbf{Benchmarks.}
Our NNCS benchmark suite consists of Benchmarks 1 - 6 from~\cite{Dutta_Others__2019__Reachability,huang2019reachnn}, Discrete Mountain Car (MC) from~\cite{ivanov2018verisig}, Adaptive Cruise Control (ACC) from~\cite{tran2020nnv}, 2D Spacecraft Docking from~\cite{ravaioli2022safe}, Attitude Control from~\cite{huang2022polar}, Quadrotor-MPC from~\cite{ivanov2018verisig} and QUAD20 from~\cite{huang2022polar}. The performance of a tool is evaluated based on proving the reachability of tasks defined for the benchmarks. They are defined as follows.

\noindent\textit{Benchmark 1-6.}
The reachability verification task asks if the NNCS can reach the target set from any initial state in $N$ control steps. We give the definitions for the initial and target sets in Table~\ref{tab:benchmarks1_6}.


\begin{table}[t]
    \centering
    \caption{Reachability verification tasks for Benchmark 1-6.}
    \label{tab:benchmarks1_6}
    \begin{adjustbox}{width=.98\columnwidth}
        \begin{tabular}{@{}c|c|c|c@{}}
            \toprule
            Benchmark & Initial set & Target set\tablefootnote{The target set is defined on the first two dimensions by default.} & $N$ \\
            \midrule
            \midrule
            1 & $[0.8, 0.9]\times [0.5, 0.6]$ & $[0, 0.2] \times [0.05, 0.3]$ & 35  \\
            \midrule
            2 & $[0.7, 0.9] \times [0.7, 0.9]$  & $[-0.3, 0.1] \times [-0.35, 0.5]$ & 10 \\
            \midrule
            3 & $[0.8, 0.9] \times [0.4, 0.5]$ & $[0.2, 0.3] \times [-0.3, -0.05]$ & 60 \\
            \midrule
            4 & \stackanchor{$[0.25, 2.27] \times [0.08, 0.1] $} {$ \times [0.25, 0.27]$} & $[-0.2, -0.1] \times [0, 0.05]$ & 10 \\
            \midrule
            5 & \stackanchor{$[0.38, 0.4] \times [0.45, 0.47]$} {$\times [0.25, 0.27]$}  & $[-0.43, -0.15] \times [0.05, 0.22]$ & 10 \\
            \midrule
            6 & \stackanchor{$[-0.77, -0.75]\times [-0.45, -0.43]$}{$\times [0.51, 0.54] \times [-0.3, -0.28]$} & $[-0.1, 0.2]\times [-0.9, -0.6]$ & 10 \\
            \bottomrule
        \end{tabular}
    \end{adjustbox}
\end{table}

\noindent\textit{Discrete-Time Mountain Car (MC).} It is a 2-dimensional NNCS describing an under-powered car driving up a steep hill. 
We consider the initial condition defined by $x_0 \in [-0.53, -0.5]$ and $x_1 = 0$. The target is $x_0 \ge 0.2$ and $x_1 \ge 0$ where the car reaches the top of the hill and is moving forward. The total control steps $N$ is 150.

\noindent\textit{Adaptive Cruise Control (ACC).} 
The benchmark models the moves of a lead vehicle and an ego vehicle. The NN controller tries to maintain a safe distance between them. We use the definition of the initial and target set given in~\cite{ivanov2020verifying}, and the number of control steps is $N=50$.


\noindent\textit{2D Spacecraft Docking (Example~\ref{exam:docking}).}
The initial set is defined by $x, y\in[24, 26]$, $v_x=v_y= -0.1378$, and $v_{safe}\in [0.2697, 0.2755]$ which is directly computed based on the ranges of $x, y$. The total control steps $N$ is $120$.
In this benchmark, we only verify the safety property.  That is, the NN controller should maintain $\sqrt{v_x^2 + v_y^2}\leq v_{safe}$ all the time. 

\noindent\textit{Attitude Control \& QUAD20.}
The reachability problems for the two benchmarks are the same as the ones given in~\cite{huang2022polar}.

\noindent\textit{Quadrotor-MPC.}
This benchmark is originally given in~\cite{ivanov2020verifying}. It consists of a quadrotor and a planner. The position of the quadrotor is indicated by the state variables $(p_x, p_y, p_z)$, while the velocity in the 3 dimensions is given by $(v_x, v_y, v_z)$. The velocity of the planner is $(b_x, b_y, b_z)$ which has a piecewise-constant definition: $(b_x, b_y, b_z) = (-0.25, 0.25, 0)$ for $t\in[0,2]$ (first 10 steps), $(b_x, b_y, b_z) = (-0.25, -0.25, 0)$ for $t\in[2,4]$, $(b_x, b_y, b_z) = (0, 0.25, 0)$ for $t\in[4,5]$, and $(b_x, b_y, b_z) = (0.25, -0.25, 0)$ for $t\in[5,6]$. The control input $\theta$, $\phi$, $\tau$ is determined by an NN ``bang-bang" controller, which is a classifier mapping system states to a finite set of control actions. The initial set is defined as $p_x - q_x \in [-0.05, -0.025]$, $p_y - q_y \in [-0.025, 0]$, and $p_z - q_z = v_x = v_y = v_z = 0$. The verification task asks to prove that all reachable states in 30 control steps should be in the safe set $-0.32 \leq p_x - q_x, p_y - q_y, p_z - q_z \leq 0.32$.

\noindent\textbf{Evaluation Metrics.} The tools are compared based on their performance on all of the benchmarks. Since tools are using different hyper-parameters, we tune their settings for each benchmark, try to make them produce similar reachable set over-approximations and only compare the time costs. For a tool that is not able to handle a benchmark, we present its result with the best setting that we can find.

\noindent\textbf{Stopping Criteria.}
We stop the run of a tool when the reachability problem is proved, or the tool raises an error or is terminated by the operating system due to a runtime system error such as being out of memory.

Hence, every test produces one of the following four results: \textbf{(Yes)} the reachability property is proved, \textbf{(No)} the reachability property is disproved, \textbf{(U)} the computed over-approximation is too large to prove or disprove the property with the best tool setting we can find, \textbf{(DNF)} a tool or system error is reported and the reachability computation fails.

\noindent\textbf{Experimental Results.}
Table~\ref{tab:experiments}, Fig.~\ref{fig:all_benchmarks}.  Because Quadrotor-MPC is a hybrid system, only POLAR-Express, Verisig 2.0, and $\alpha,\beta$-CROWN\,+\,Flow* are able to deal with it. It is verified to be safe by POLAR-Express with 13.1 seconds and by Verisig 2.0 with 961.4 seconds. The result is Unknown for $\alpha,\beta$-CROWN\,+\,Flow* after 10854 seconds. 




\begin{table*}[pt!]
\caption[Caption for LOF]{Verification results. \textit{Dim} indicates the dimensions of input for NN controllers in each benchmark. NN Architecture lists activation functions and network structures in each benchmark. For example, ReLU ($n \times k$) represents that the network has $n$ hidden layers and $k$ neurons per layer. Hidden layers and the output layer use ReLU as an activation function. ReLU\_tanh represents that hidden layers use ReLU and the output layer uses tanh. For each tool, we provide verification results and time in seconds. If a property could not be verified, it is marked as U (Unknown). If a tool crashed on a benchmark with an internal error, it is marked as DNF. mt is short for multi-threads (we used 12 threads) and st is short for single-thread. The runtime of POLAR-Express is decomposed into the running time of the propagation of TMs in the neural-network controller (multi-threading support) and separately the reachability computation for the dynamical system. For neural-network controller that does not have any ReLU activation function, the single-threaded version of POLAR-Express is exactly the same as our prior prototype POLAR~\cite{huang2022polar}.}\label{tab:experiments}
\begin{adjustbox}{width=2\columnwidth}
\begin{threeparttable}
\begin{tabular}{@{}c|c|c|c|c|c|c|c|c|c|c@{}}
\toprule
Benchmarks       & \textit{Dim}     & NN Architecture      & POLAR-Express (mt)        & POLAR-Express (st) & Verisig 2.0 (mt)  & $\alpha,\beta$-CROWN\,+\,Flow* & NNV  & JuliaReach     & CORA    & RINO \\ \midrule \midrule
\multirow{4}{*}{Benchmark 1} & \multirow{4}{*}{2} & ReLU (2 x 20)     & Yes (0.54 + 0.08) & \textbf{Yes (0.09 + 0.08)}  & --         & U (230.1)   & DNF       & Yes (19.7)  & Yes (6.7)   & --\\
                             &                    & sigmoid (2 x 20)  & \textbf{Yes (0.77 + 0.09) }  & Yes (1.0 + 0.09) & Yes (26.6) & U (506.8)   & U (14.0)  & Yes (16.9)  & Yes (6.7)  & Yes (1.0)\\
                             &                    & tanh (2 x 20)     & \textbf{Yes (1.44 + 0.18)}   & Yes (2.6 + 0.14) & Yes (62.1) & U (299.3)   & U (18.6)  & Yes (17.2)  & Yes (3.2)  & Yes (1.9)\\
                             &                  & ReLU\_tanh (2 x 20) & Yes (0.55 + 0.08)  &\textbf{ Yes (0.1 + 0.07)}  & --         & U (293.5)   & U (18.0)  & Yes (17.2)  & Yes (4.6)  & --\\ \midrule
\multirow{4}{*}{Benchmark 2} & \multirow{4}{*}{2} & ReLU (2 x 20)     & Yes (0.13 + 0.02)  &\textbf{ Yes (0.02 + 0.01)} & --         & U (112.9)    & U (13.6)  & U (18.8)    & Yes (4.0)  & --\\
                             &                    & sigmoid (2 x 20)  & Yes (0.37 + 0.17)  & Yes (0.6 + 0.15)           & Yes (6.5)  & U (134.5)    & U (6.8)   & U (18.9)    & Yes (1.1)   & \textbf{Yes (0.3)}\\
                             &                    & tanh (2 x 20)     & \textbf{Yes (1.3 + 0.5)} & Yes (2.7 + 0.5)   & U (4.7)    & U (161.1)    & U (5.9)   & U (18.9)    & DNF        & DNF\\
                             &                    & ReLU\_tanh (2 x 20)&Yes (0.26 + 0.5)  & \textbf{Yes (0.1 + 0.4)}  & --         & U (85.5)    & U (17.1)  & U (18.9)    & Yes (1.0)  & --\\ \midrule
\multirow{4}{*}{Benchmark 3} & \multirow{4}{*}{2} & ReLU (2 x 20)     & Yes (1.2 + 0.85)  & \textbf{Yes (0.4 + 0.75)}  & --         & U (217.8)   & U (21.0)  & Yes (17.6)  & Yes (4.5)  & --\\
                             &                    & sigmoid (2 x 20)  & Yes (3.3 + 0.8)  & Yes (7.1 + 0.8)           & Yes (38.2) & U (353.0)   & U (26.8)  & Yes (17.5)  & \textbf{Yes (2.6)}  & Yes (2.8) \\
                             &                    & tanh (2 x 20)     & Yes (3.2 + 0.8)  & Yes (6.2 + 0.7)           & Yes (31.9) & U (424.6)   & U (27.1)  & Yes (17.7)  & \textbf{Yes (2.4)}  & Yes (6.5)\\
                             &                    & ReLU\_sigmoid (2 x 20)&No (1.4 + 0.8)&\textbf{No (0.6 + 0.7)}    & --         & U (141.2)    & U (20.6)  & No (17.7)   & No (2.2)   & --\\ \midrule
\multirow{4}{*}{Benchmark 4} & \multirow{4}{*}{3} & ReLU (2 x 20)     & Yes (0.15 + 0.02)  & \textbf{Yes (0.03 + 0.02)} & --         & U (75.4)    & DNF     & Yes (16.6)  & Yes (3.0)   & --\\
                             &                    & sigmoid (2 x 20)& No (0.24 + 0.02)    &\textbf{No (0.17 + 0.02)}& No (5.7)   & No (139.5)    & DNF     & No (16.7)   & No (0.9)    & No (0.3)\\
                             &                    & tanh (2 x 20)     & No (0.23 + 0.02)   & No (0.17 + 0.02)            & No (4.6)   & No (160.9)    & DNF     & No (16.9)   & No (0.9)    & \textbf{No (0.1)}\\
                             &                    & ReLU\_tanh (2 x 20) & Yes (0.14 + 0.02)&\textbf{Yes (0.03 + 0.02)}  & --         & U (89.6)    & U (5.4)   & Yes (16.9)  & Yes (1.0)   & --\\ \midrule
\multirow{4}{*}{Benchmark 5} & \multirow{4}{*}{3} & ReLU (3 x 100)    & Yes (2.4 + 0.02)  & \textbf{Yes (1.8 + 0.02)}  & --         & U (111.4)    & DNF     & U (18.6)    & Yes (3.4)  &--\\
                             &                    & sigmoid (3 x 100) & No (4.2 + 0.02)   & \textbf{No (3.7 + 0.02)}   & No (93.1)  & U (229.3)    & DNF     & U (19.0)    & U (1.3)     & U (7.2)\\
                             &                    & tanh (3 x 100)    & Yes (4.2 + 0.02)  &  Yes (3.8 + 0.02)  & Yes (74.2) & U (263.0)   & U (125.0) & U (18.7)    & \textbf{Yes (1.2) } & Yes (2.7)\\
                             &                    & ReLU\_tanh (3 x 100)& Yes (2.5 + 0.02)&  Yes (2.1 + 0.02)  & --         & U (96.1)    & U (4.1)   & U (18.8)    & \textbf{Yes (1.7)}  &--\\ \midrule
\multirow{4}{*}{Benchmark 6} & \multirow{4}{*}{4} & ReLU (3 x 20)     & Yes (0.2 + 0.05)  & \textbf{Yes (0.05 + 0.05)}  & --         & U (149.8)    & U (8.8)   & Yes (19.6)  & Yes (8.6)   &--\\
                             &                    & sigmoid (3 x 20)  & Yes \textbf{(0.47 + 0.05)}  & Yes (0.7 + 0.05)           & Yes (24.3) & U (217.5)    & U (9.9)   & Yes (19.5)  & Yes (2.7)   & Yes (0.6) \\
                             &                    & tanh (3 x 20)     & Yes (0.5 + 0.05)  & Yes (0.7 + 0.05)           & Yes (18.8) & U (259.1)   & U (9.1)   & Yes (20.4)  & Yes (3.7)    & \textbf{Yes (0.5)} \\
                             &                    & ReLU\_tanh (3 x 20) & Yes (0.2 + 0.05)&\textbf{Yes (0.05 + 0.05)}   & --         & U (109.3)    & U (5.9)   & Yes (19.6)  & Yes (5.1)   &--\\ \midrule
Discrete Mountain Car        & 2                  & sigmoid\_tanh (2 x 16)&\textbf{Yes (3.0 + 0.12)}& Yes (4.0 + 0.14) & Yes (70.3) & U (450.0)   & U (18.1)  & --         &    --        & DNF \\ \midrule
ACC                          & 6                  & tanh (3 x 20)    & \textbf{Yes (2.6 + 0.15)} & Yes (4.4 + 0.14)    & Yes (1980.2) & U (1806.4)  & U (39.6)  & U (29.9)   & Yes (3.7)    & Yes (10.1)\\ \midrule
2D Spacecraft Docking        & 6                  & tanh (2 x 64)    & Yes (15.8 + 29.7) & Yes (25.5 + 29.8)           & U (3331.0)   & U (3119.7)  & U (74.0)  & Yes (37.3) & \textbf{Yes (33.5)} & DNF \\ \midrule
Attitude Control             & 6                  & sigmoid (3 x 64) & Yes (8.0 + 0.6)  & Yes (10.1 + 0.5)& U (1271.5)   & U (582.6)  & U (358.2)  & Yes (19.9) & \textbf{Yes (1.9)}  & Yes (121.3) \\ \midrule
QUAD20                       &   12               &  sigmoid (3 x 64)& \textbf{Yes (316.8 + 212.5)} & Yes (809.4 + 197.6)& U (17939.7, 1step)    & U (370.2, 3step)    & U (325.5, 1step)  & U (60.6, 25step)    & U (2219, soundness)      & DNF\\ \bottomrule
\end{tabular}
\end{threeparttable}
\end{adjustbox}
\end{table*}

\begin{figure*}[hbp!]
\captionsetup[subfloat]{farskip=2pt,captionskip=1pt}
	\centering	
	\subfloat[][]{%
		\includegraphics[width=0.33\textwidth]{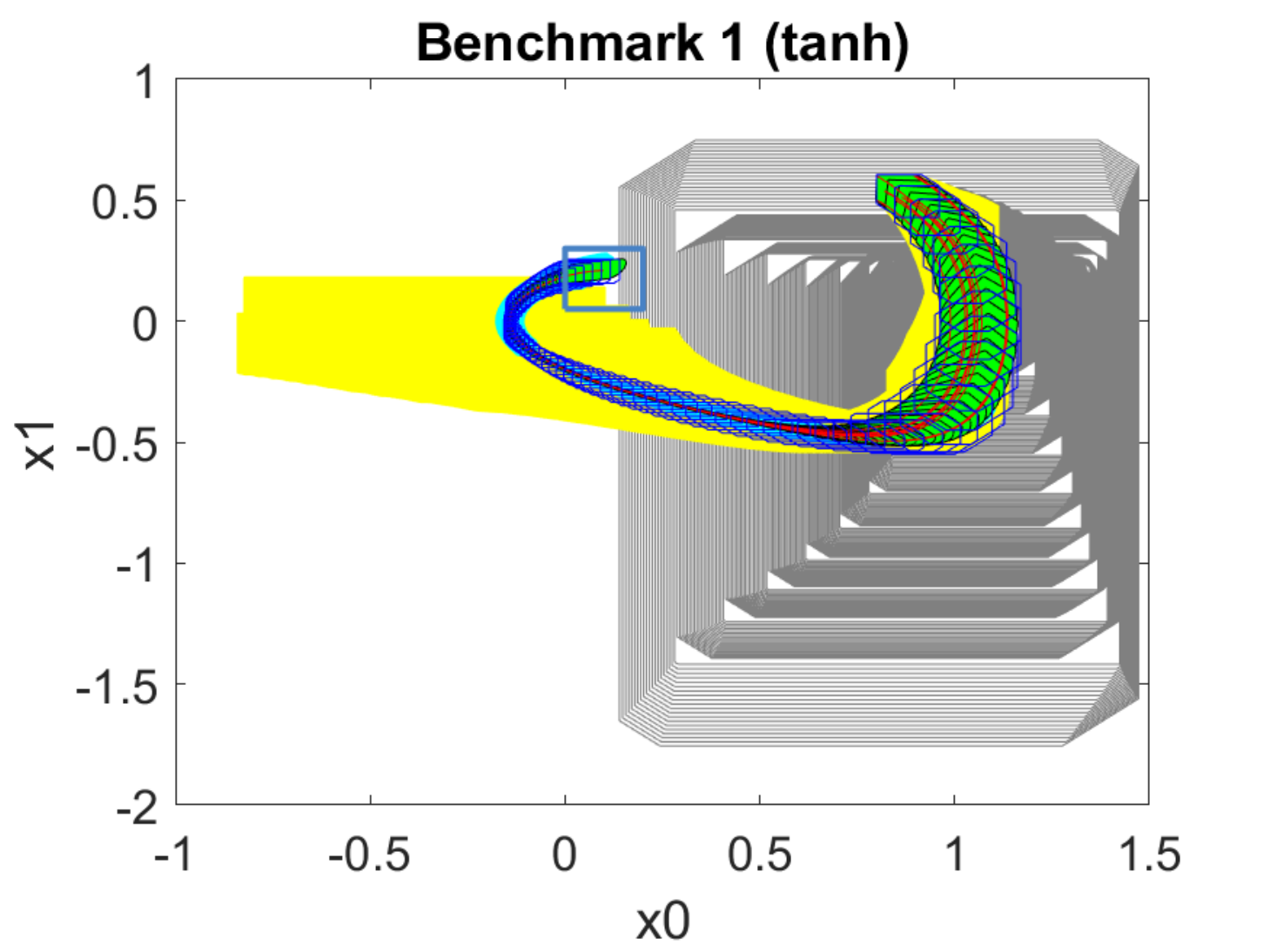}%
	} 
	\subfloat[][]{%
		\includegraphics[width=0.33\textwidth]{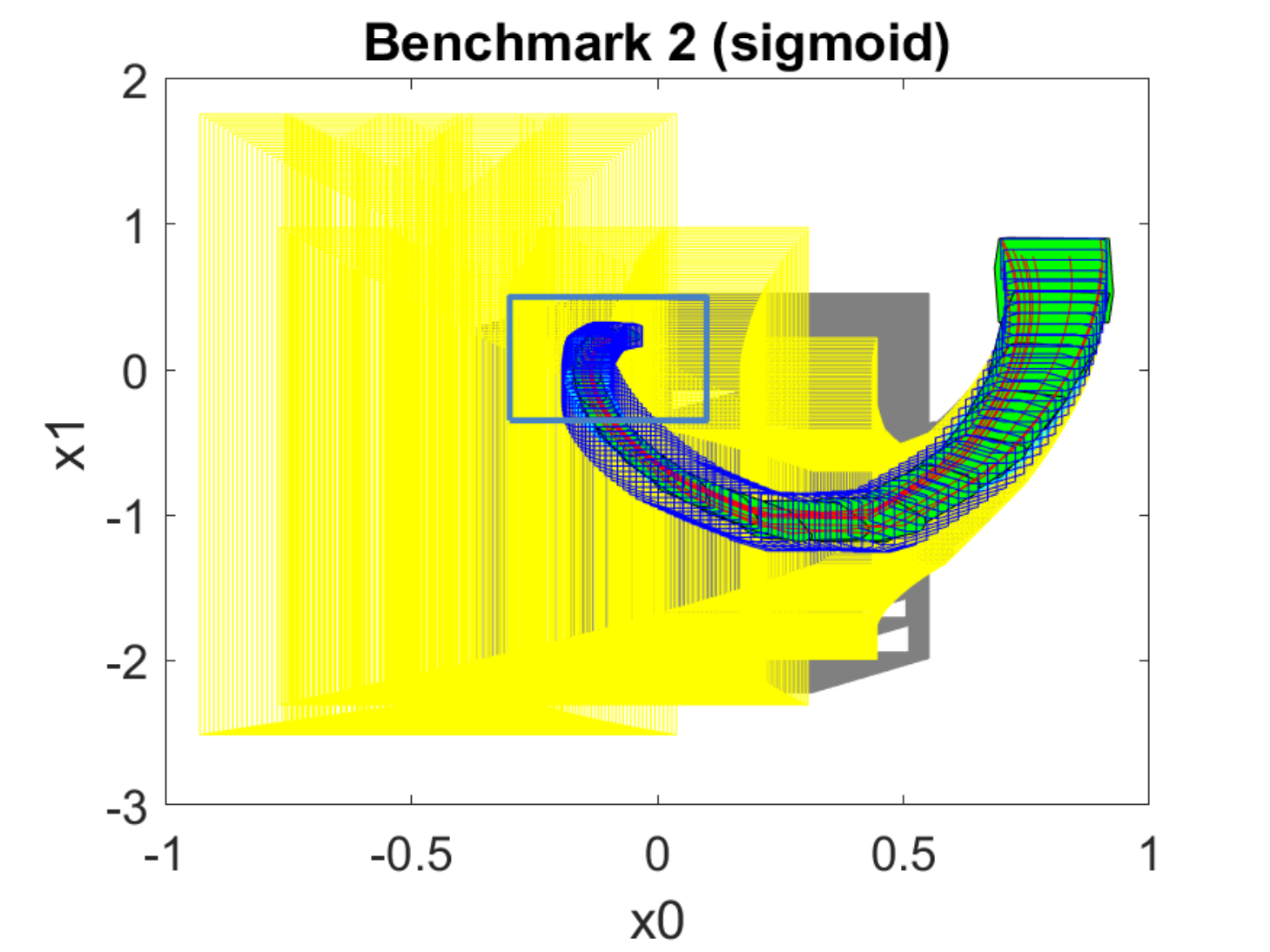}%
	} 
	\subfloat[][]{%
		\includegraphics[width=0.33\textwidth]{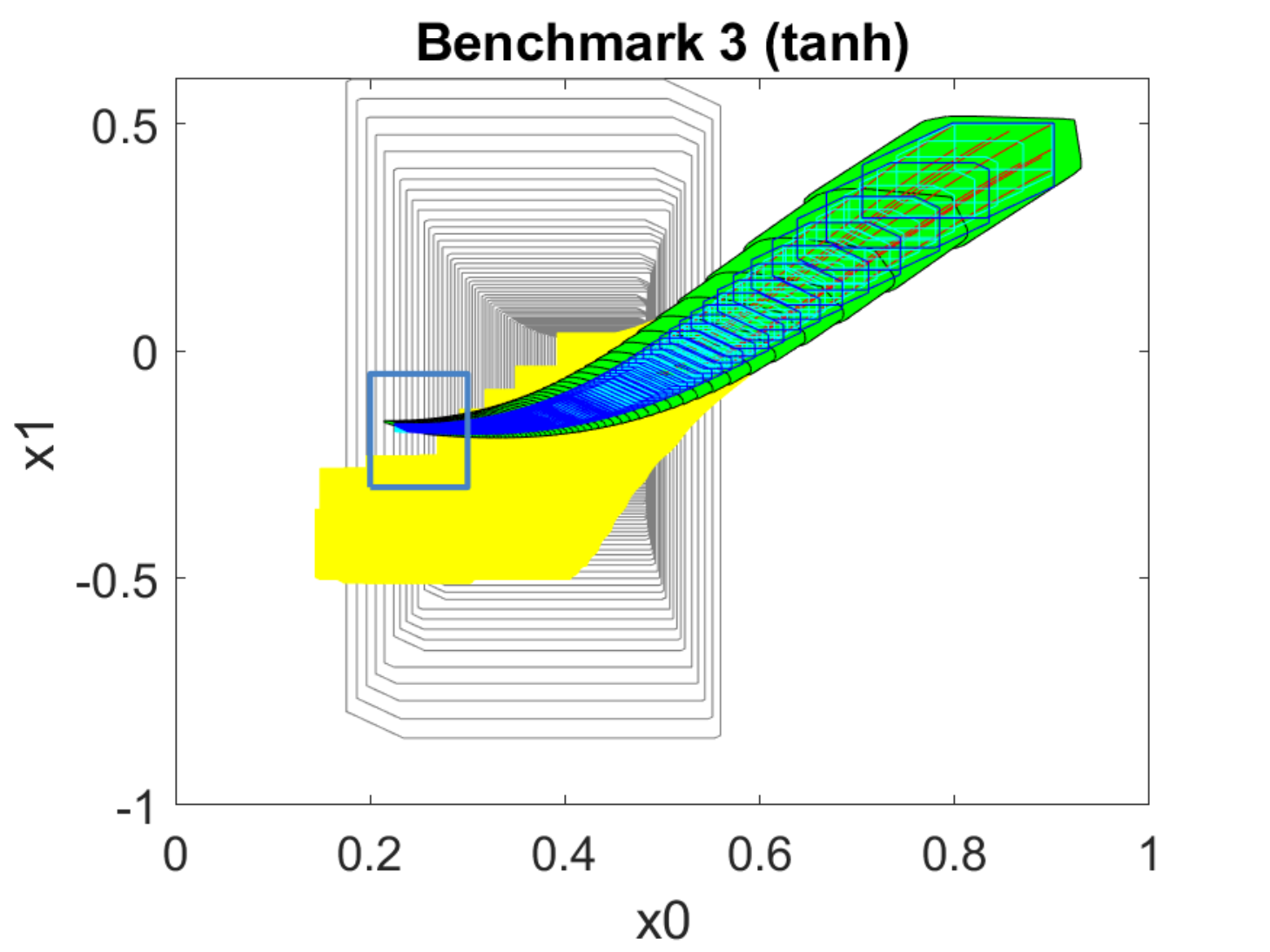}%
	} \
	\subfloat[][]{%
		\includegraphics[width=0.33\textwidth]{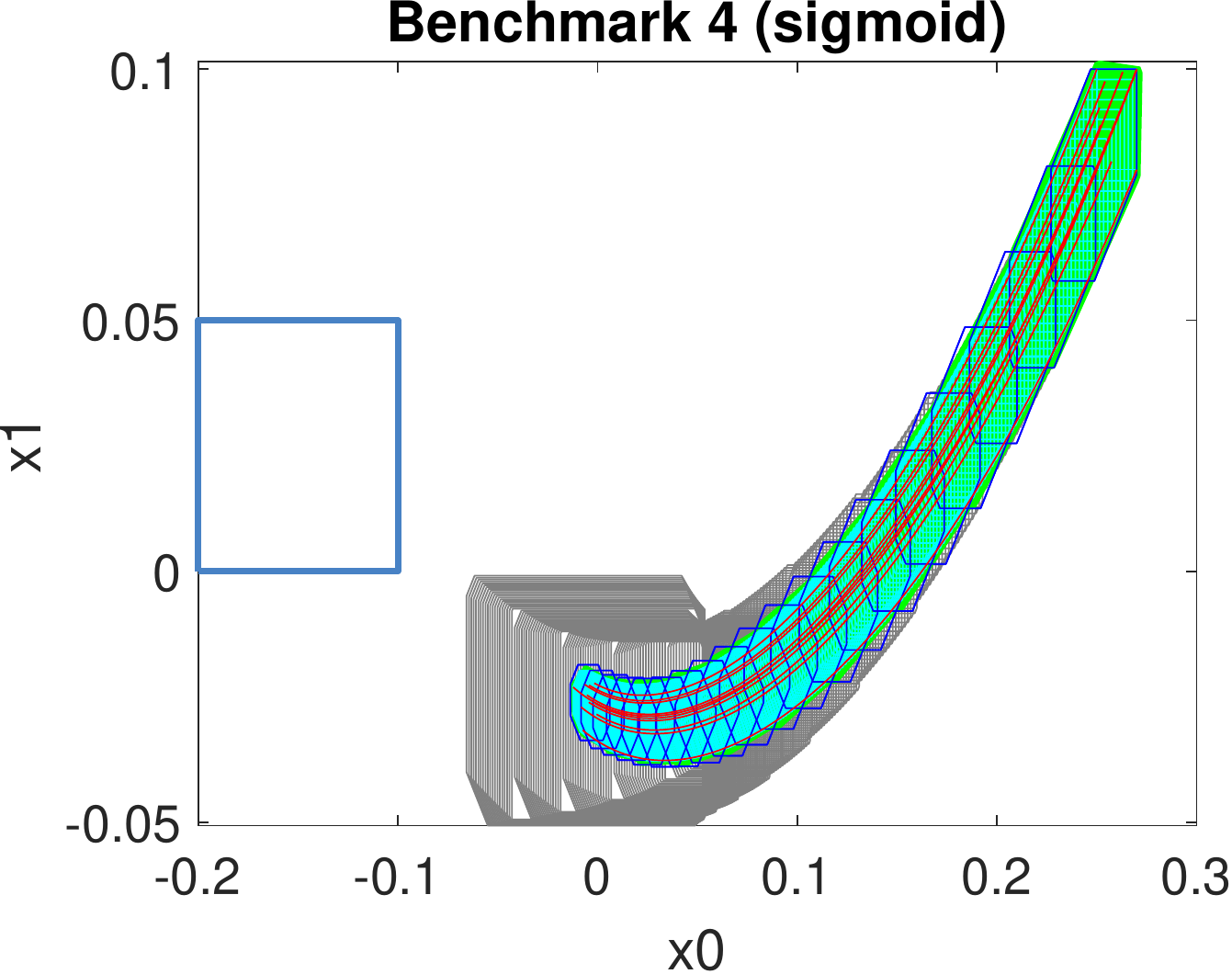}%
	}
	\subfloat[][]{%
		\includegraphics[width=0.33\textwidth]{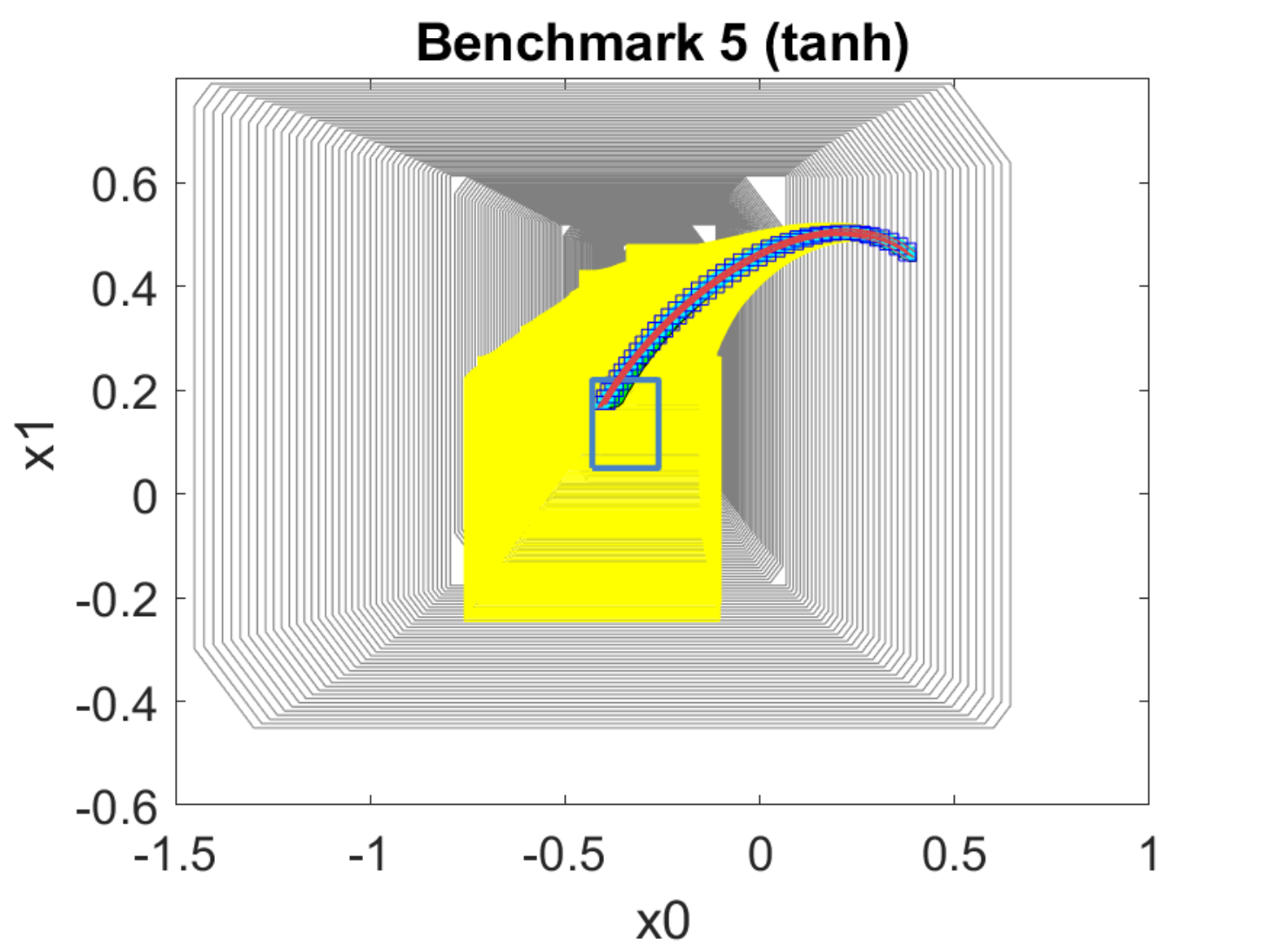}%
	}
	\subfloat[][]{%
		\includegraphics[width=0.33\textwidth]{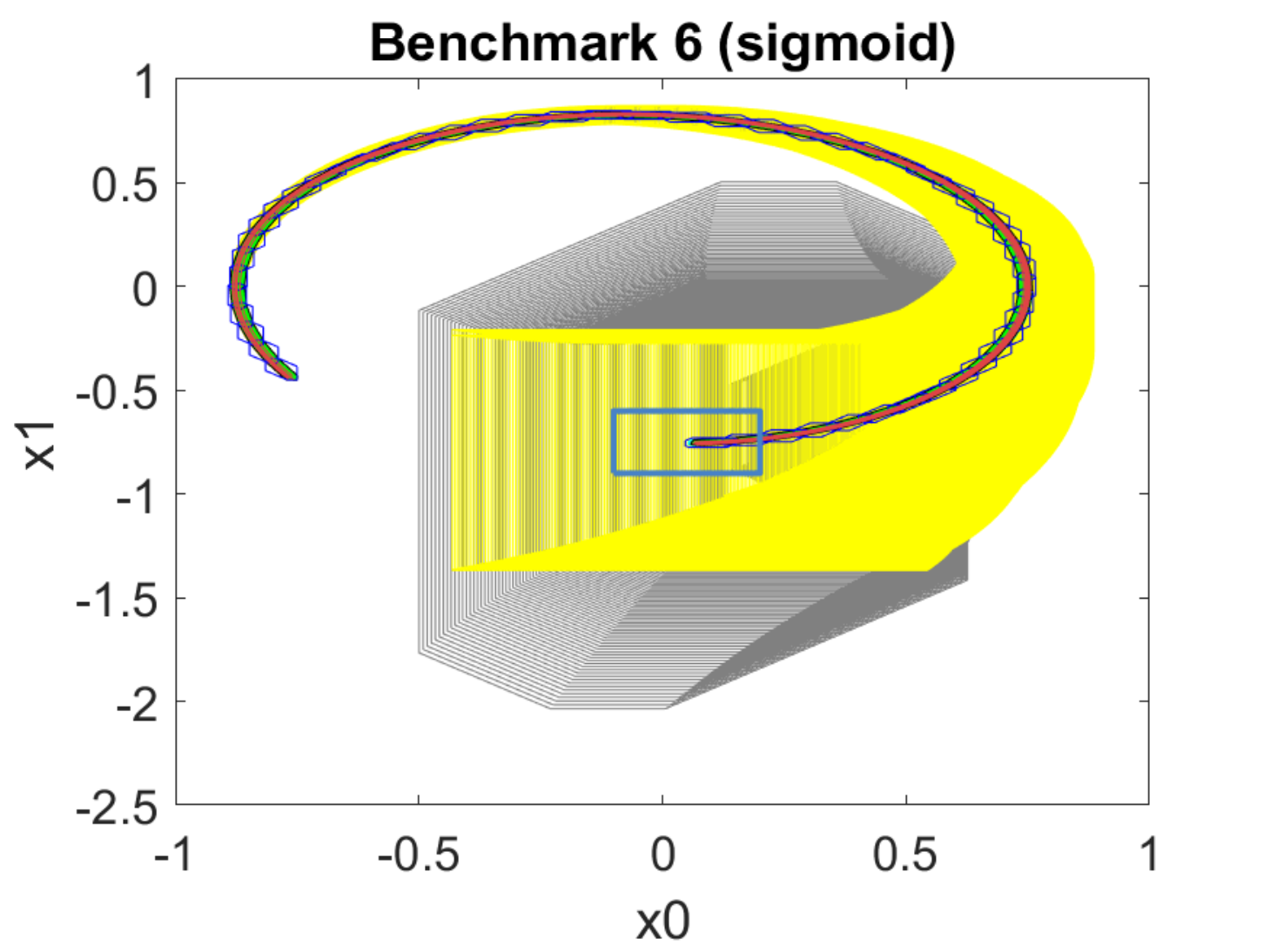}%
	} \ 
	\subfloat[][]{%
		\includegraphics[width=0.33\textwidth]{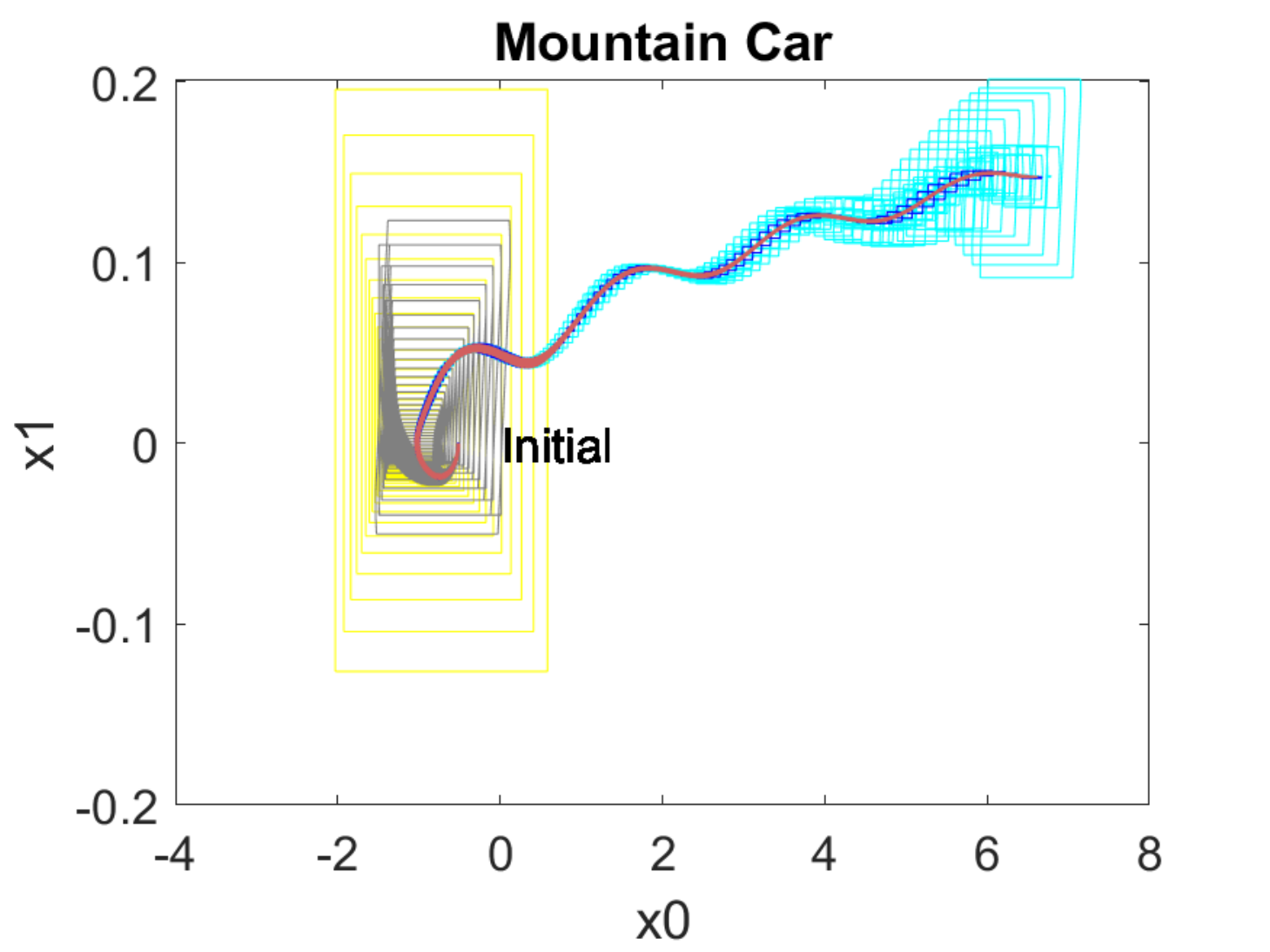}%
	} 
 	\subfloat[][]{%
		\includegraphics[width=0.33\textwidth]{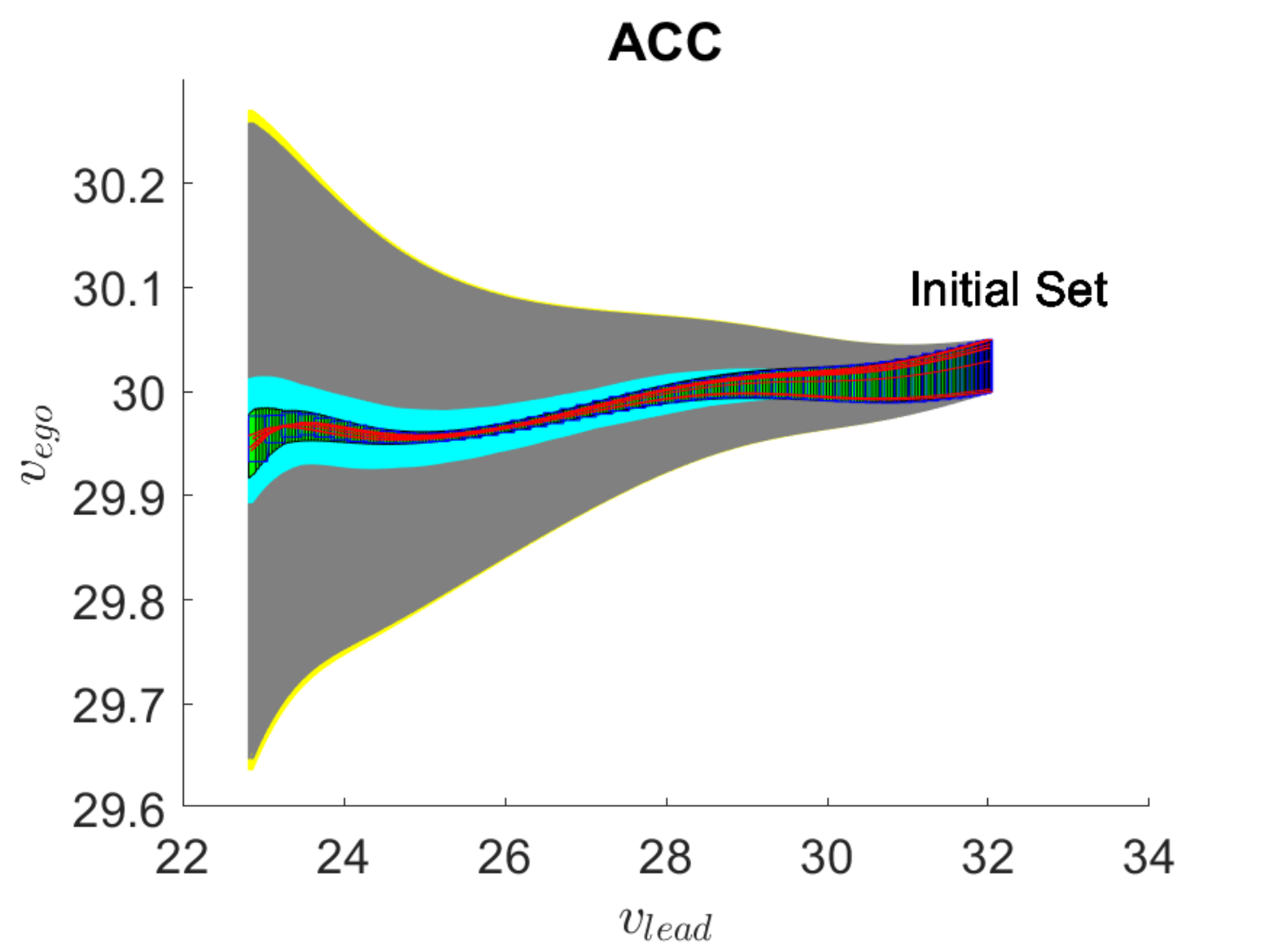}%
	}
	\subfloat[][]{%
		\includegraphics[width=0.33\textwidth]{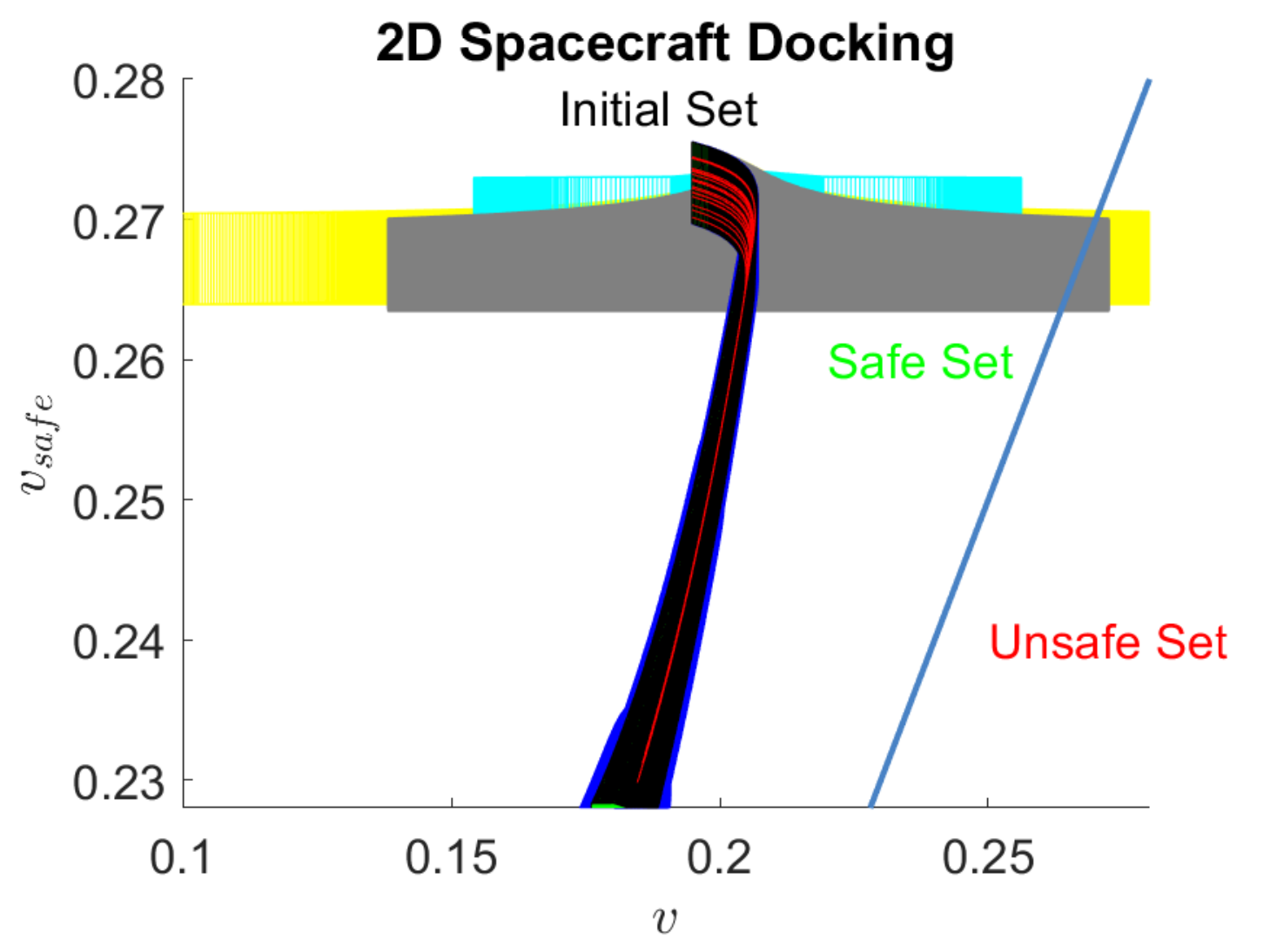}%
	} \ 
 	\subfloat[][]{%
		\includegraphics[width=0.33\textwidth]{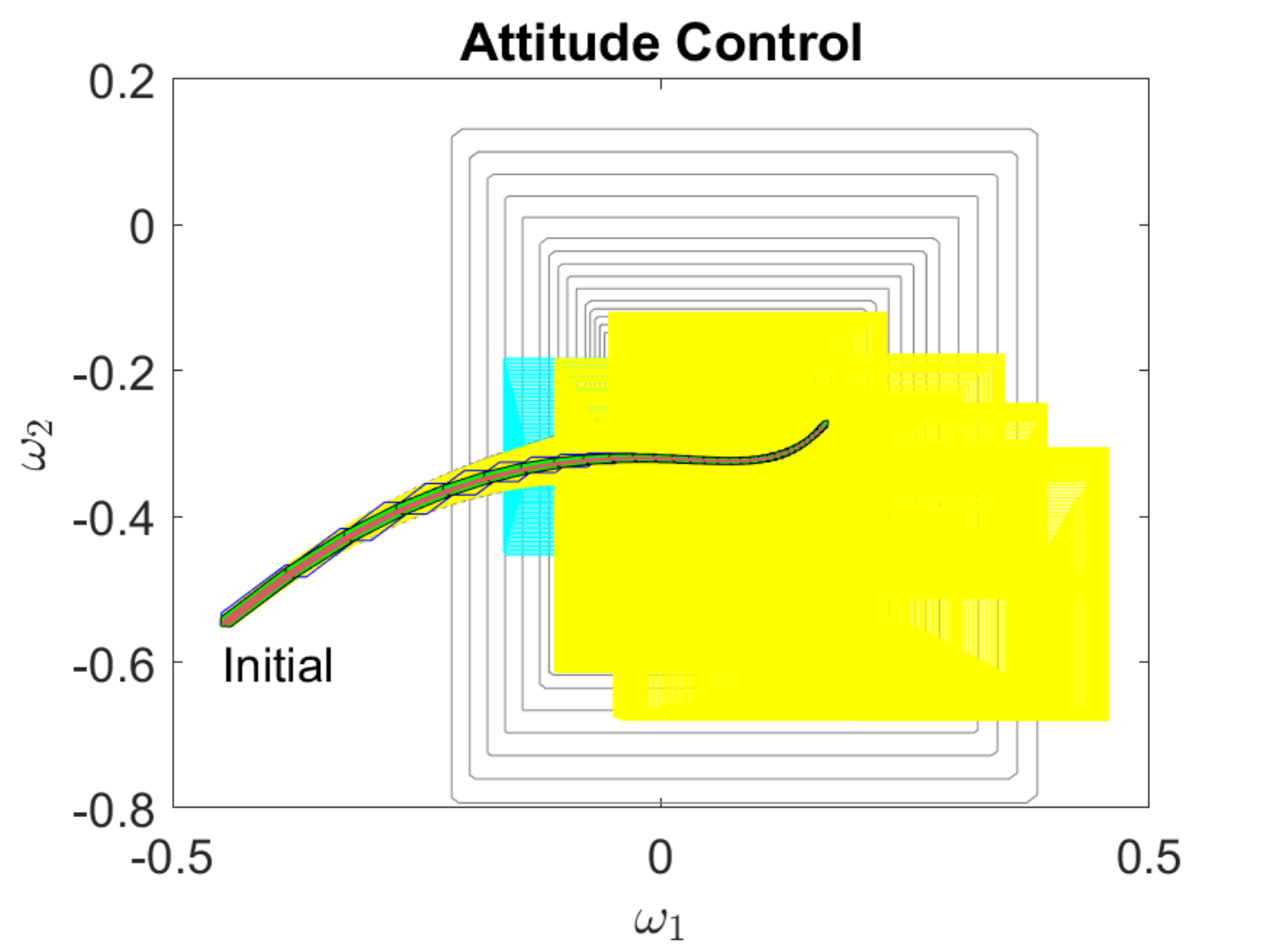}%
	}
 	\subfloat[][]{%
		\includegraphics[width=0.33\textwidth]{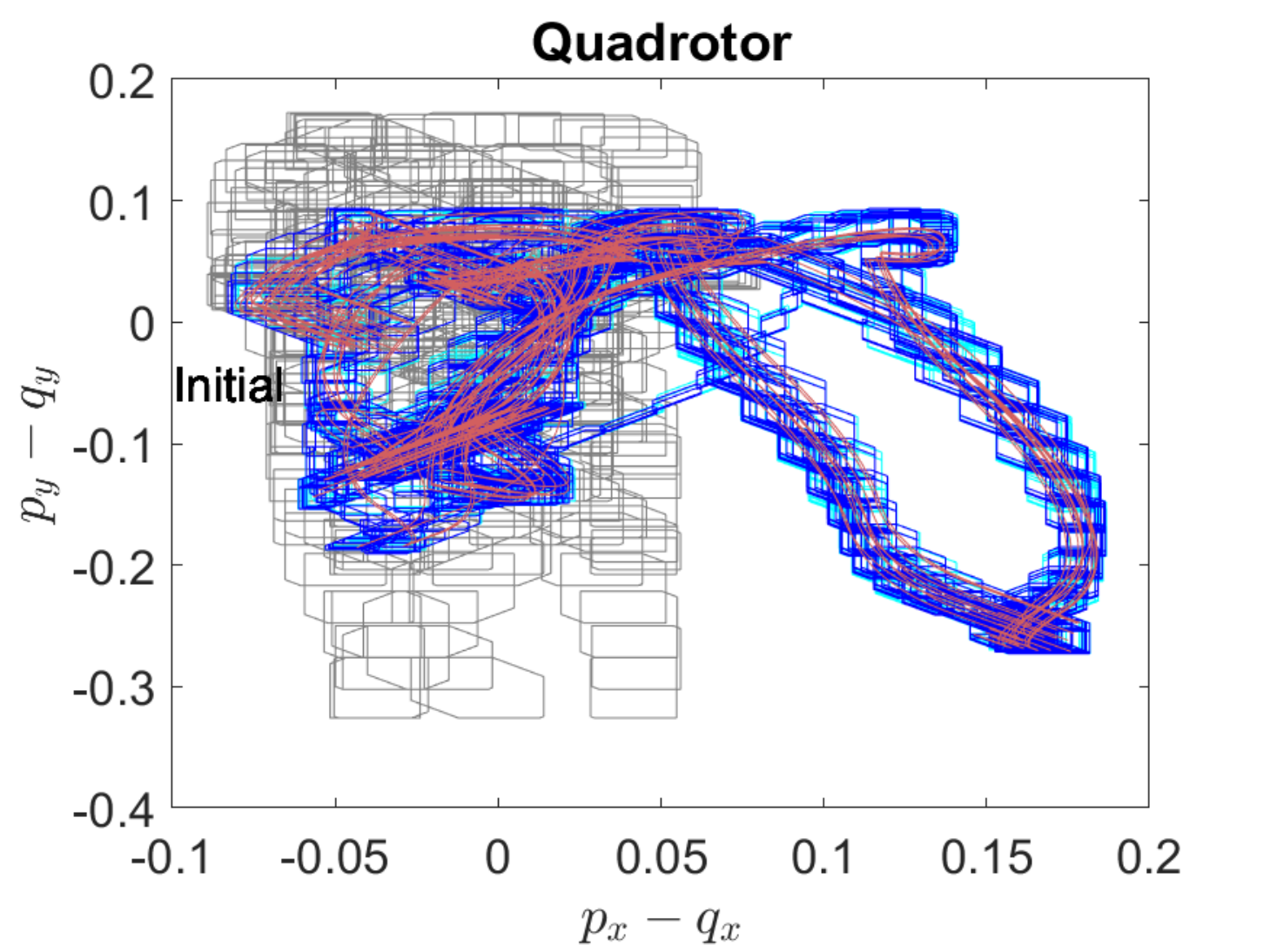}%
	}
 	\subfloat[][]{%
		\includegraphics[width=0.33\textwidth]{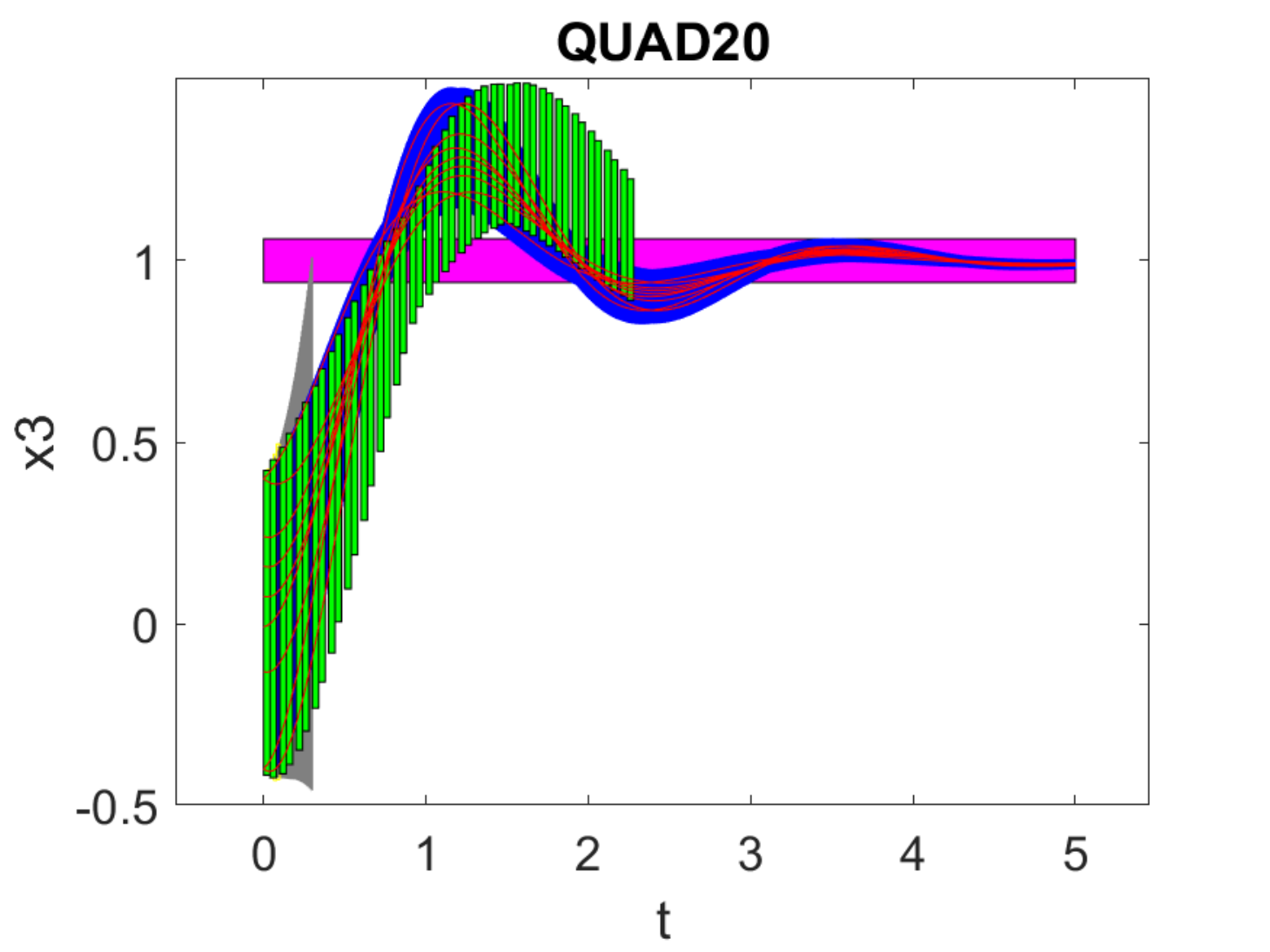}%
	}
	\caption{Computed reachable sets of all examples by different tools.  POLAR-Express (blue sets), Verisig 2.0 (cyan sets), $\alpha,\beta$-CROWN\,+\,Flow*  (grey sets), NNV (yellow sets), CORA (green sets) and simulation traces (red curves), JuliaReach and RINO have their own plotting support and are not shown in this paper. }
	\label{fig:all_benchmarks}
\end{figure*}

\subsection{Challenges with Running Other Tools}
We found soundness issues when running the CORA tool for Benchmark 5 and  the QUAD20 example. In Benchmark 5, both simulation traces and reachable sets of CORA deviate from the other tools. In QUAD20, the reachable sets computed by CORA cannot cover the simulation traces (i.e. not an over-approximation), shown in Fig~\ref{fig:all_benchmarks} (l). 

The default setup in JuliaReach reports the runtime of running the \textit{same} reachable set computation the second time after a ``warm-up" run (whose runtime is not included), 
likely to take advantage of cache effects or saved computations. 
For a fair comparison with the other tools, we report the runtime of the first run.

RINO has three ``DNF"s caused by division-by-zero errors. Both RINO and JuliaReach have their own plot functions, while the other tools plot reachable sets in MATLAB. This is an issue in several examples such as the plot function in RINO taking too long to plot the reachable sets.

For the most complicated QUAD20 example, Verisig 2.0 failed after 17939 seconds during 1-step reachable set computation, $\alpha,\beta$-CROWN\,+\,Flow* failed after 3 control steps, NNV failed after 1 step, CORA has the soundness issue as mentioned earlier, the reachable set of JuliaReach explodes after 25 control steps, and RINO has the division-by-zero error.


\subsection{Experimental Comparison and Discussion}

According to the experimental results in Table~\ref{tab:experiments} and Fig~\ref{fig:all_benchmarks}, POLAR-Express can verify all the benchmarks and achieve the overall best performance in terms of the tightness of the reachable set computations and runtime efficiency. 
We can also observe that the multi-threading support for POLAR-Express helps to reduce runtime in higher dimensional systems. However, because of the overhead involved, it may take longer than the single-threaded version for some of the lower dimensional benchmarks.
\zhou{
Thus, we believe that multi-threading support will be more suitable for higher dimensional tasks.}

Verisig 2.0 can only handle neural-network controllers with differential activation functions (e.g., sigmoid and tanh). In the lower-dimensional benchmarks, the reachable sets computed by POLAR-Express and Verisig 2.0 almost overlap with each other. However, POLAR-Express takes much less time to compute the reachable sets in all cases. In the higher-dimensional benchmarks, compared with POLAR-Express, Verisig 2.0 either has much larger over-approximation errors (e.g., Figure~\ref{fig:all_benchmarks} (g) (h) (k)) or can only compute fewer steps of reachable sets (e.g., Figures~\ref{fig:all_benchmarks} (i) (j) (l)). 

NNV and $\alpha,\beta$-CROWN are pure (symbolic) bound estimation techniques. As explained in Section~\ref{sec:related} and can be observed in Fig~\ref{fig:all_benchmarks}, these methods can produce large over-approximation errors when used in reachable set computations, even if the bound estimations are relatively tight compared with function-over-approximation-based approaches. 
Due to the closed-loop nature of NNCSs, a large over-approximation will also slow down the reachable set computations for the subsequent control steps, as evident in the runtime of $\alpha,\beta$-CROWN$\,+\,$Flow* in Table~\ref{tab:experiments} even though $\alpha,\beta$-CROWN itself is quite efficient.

JuliaReach's runtimes are close across many benchmarks despite the substantial differences in the system dynamics and the neural networks in those benchmarks. After breaking down those runtimes, we notice that memory allocation always constitutes the largest portion of the runtimes which makes the difference of benchmarks less significant.
On the other hand, it fails to verify Benchmark 2 and Benchmark 5
due to large over-approximation errors. 
CORA can achieve similar runtime efficiency and over-approximation tightness as POLAR-Express in some benchmarks, but it suffers from soundness issues (manifested on Benchmark 5 and QUAD20) as noted earlier. 
RINO can only handle neural networks with differential activation functions such as sigmoid and tanh. It is quite efficient on the lower-dimensional systems but is slow for the higher-dimensional systems. In addition, the tool contains division-by-zero bugs that resulted in the reported DNF errors.

\smallskip
\noindent

\section{Conclusion and Future Work} 
\label{sec:conclusion}

{We present POLAR-Express, a formal reachability verification tool for NNCSs, which uses layer-by-layer propagation of TMs to compute function over-approximations of neural-network controllers. We provide a comprehensive comparison of POLAR-Express with existing tools and show that POLAR-Express can achieve state-of-the-art efficiency and tightness in reachable set computations.
On the other hand, current techniques still do not scale well to high-dimensional cases. In our experiment, the performance of Verisig 2.0 degrades significantly for the 6-dimensional examples, and POLAR-Express is also less efficient in the QUAD20 example. We believe state dimensions, control step sizes, and the numbers of total control steps are the key factors in scalability. 
As TMs are parameterized by state variables, higher state dimensions will lead to a more tedious polynomial expression in the TMs. 
Meanwhile, a large control step or a large number of total control steps can make it more difficult to propagate the state dependencies across the plant dynamics and across multiple control steps. We believe that addressing these scalability issues
will be the main subject of future work in NNCS reachability analysis.


\smallskip
\noindent\textbf{Acknowledgement}. We gratefully acknowledge the support from the National Sci- ence Foundation awards CCF-1646497, CCF-1834324, CNS-1834701, CNS-1839511, IIS-1724341, CNS-2038853, ONR grant N00014-19-1-2496, and the US Air Force Research Laboratory (AFRL) under contract number FA8650-16-C-2642.


\bibliographystyle{IEEEtran}
\bibliography{IEEEabrv,./chao,./jiameng,./xin,./weichao,./zhilu}

\end{document}


\title{ReachNN: Reachability Analysis of Neural-Network\\ Controlled Systems -- Appendix}

\maketitle

\begin{lemma}\label{lem:bernstein_bound}
 Given that $p_B(x)$ is the order $k\geq 1$ Bernstein polynomial for the function $\relu(x)$ over $x\in [a,b]$ such that $a < 0 < b$, then $\relu(x)\leq p_B(x)$ for all $a\in [a,b]$.
\end{lemma}
\begin{proof}
 The polynomial $p_B(x)$ can be expressed as
 \[
  \sum_{j=0}^k \left[\relu(a + \frac{j}{k}(b-a))\binom{k}{j}\left(\frac{x-a}{b-a}\right)^j \left(\frac{b-x}{b-a}\right)^{k-j}\right]
 \]
 then it is no smaller than $0$ that is $\relu(x)$ for all $x\in [a,0]$.
 Since
 \[
  \sum_{j=0}^k \left[(a + \frac{j}{k}(b-a))\binom{k}{j}\left(\frac{x-a}{b-a}\right)^j \left(\frac{b-x}{b-a}\right)^{k-j}\right] = x
 \]
 for all $x\in [a,b]$, and $(a + \frac{j}{k}(b-a)) \leq \relu(a + \frac{j}{k}(b-a))$, we have that $x\leq p_B(x)$ for $x\in [a,b]$, and hence $\relu(x)\leq p_B(x)$ for $x\in [0,b]$.
\end{proof}

\begin{lemma}
 Given that $p_B(x)$ is the order $k\geq 1$ Bernstein polynomial for the function $\relu(x)$ over $x\in [a,b]$ such that $a < 0 < b$, then we have that $|\relu(x) - p_B(x)|\leq p_B(0)$ for all $a\in [a,b]$.
\end{lemma}
\begin{proof}
 We first prove that the lemma holds when $k = 1$. The Bernstein polynomial is defined by $p_B(x) = b \left(\frac{x-a}{b-a}\right)$ whose derivative is $\frac{b}{b-a}$ which is between $0$ and $1$. Since the derivative of $\relu(x)$ is $0$ for $x\in (a,0)$ and $1$ for $x\in (0,b)$, and by Lemma~\ref{lem:bernstein_bound}, we have that the difference $p_B(x) - \relu(x)$ is (i) positive in $x\in [a,b]$, (ii) monotonically increasing when $x\in [a,0]$, and (iii) monotonically decreasing when $x\in [0,b]$, therefore $|p_B(x) - \relu(x)|\leq p_B(0) - \relu(0) = p_B(0)$ for all $x\in [a,b]$.

 When $k \geq 2$, by Theorem 7.1.3 in \cite{Phillips/2003/polynomials}, we have that
 \[
  \begin{aligned}
  \frac{d p_B}{d x} = & k \sum_{j=0}^{k-1} \left[ \Delta\relu\left(a+\frac{j}{k}(b-a)\right)
  \binom{k-1}{j} \right. \\
  & \left.\left(\frac{x-a}{b-a}\right)^j \left(\frac{b-x}{b-a}\right)^{k-1-j}\right]
  \end{aligned}
 \]
 and
 \[
  \small
  \begin{aligned}
  \frac{d^2 p_B}{d x^2} = & \frac{k!}{(k-2)!} \sum_{j=0}^{k-2} \left[ \Delta^2\relu\left(a+\frac{j}{k}(b-a)\right) \right. \\
  & \binom{k-2}{j} \left.\left(\frac{x-a}{b-a}\right)^j \left(\frac{b-x}{b-a}\right)^{k-2-j}\right]
  \end{aligned}
 \]
 wherein
 \[
  \begin{aligned}
   & \Delta\relu\left(a+\frac{j}{k}(b-a)\right) \\
   = & \relu\left(a+\frac{j+1}{k}(b-a)\right) - \relu\left(a+\frac{j}{k}(b-a)\right)
  \end{aligned}
 \]
 \[
   \begin{aligned}
   & \Delta^2\relu\left(a+\frac{j}{k}(b-a)\right) \\
   = & \Delta\relu\left(a+\frac{j+1}{k}(b-a)\right) - \Delta\relu\left(a+\frac{j}{k}(b-a)\right)
  \end{aligned}
 \]
 We have that $0\leq dp_B/dt \leq 1$, and $d^2 p_B/dt^2 \geq 0$, since for all $k'\geq 0$, we have that
 \[
  \sum_{j=0}^{k'} \left[ \binom{k'}{j} \left(\frac{x-a}{b-a}\right)^j \left(\frac{b-x}{b-a}\right)^{k'-j}\right] = 1
 \]
 and
 \[
  \small
  \begin{aligned}
  & \relu\left(a + \frac{j+1}{k}(b-a)\right) - \relu(a + \frac{j}{k}(b-a)) \leq \frac{1}{k(b-a)}\,, \\
  & \Delta \relu\left(a + \frac{j+1}{k}(b-a)\right) \geq \Delta \relu\left(a + \frac{j}{k}(b-a)\right)
  \end{aligned}
 \]
 Similar to the case of $k=1$, we may prove that the difference $p_B(x) - \relu(x)$ still satisfies the properties (i), (ii) and (iii), and then $|p_B(x) - \relu(x)|\leq p_B(0) - \relu(0) = p_B(0)$ for all $x\in [a,b]$.
\end{proof}
